\PassOptionsToPackage{unicode}{hyperref}
\PassOptionsToPackage{hyphens}{url}
\PassOptionsToPackage{dvipsnames,svgnames,x11names}{xcolor}
\documentclass[
  12pt]{article}

\usepackage{amsmath,amssymb}

\usepackage[utf8]{inputenc} 
\usepackage[T1]{fontenc}    
\usepackage{hyperref}       
\usepackage{url}            
\usepackage{booktabs}       
\usepackage{amsfonts}       
\usepackage{nicefrac}       
\usepackage{microtype}      
\usepackage{xcolor, soul}         
\usepackage{amsmath}
\usepackage{comment}
\usepackage{amssymb}
\usepackage{amsfonts}
\usepackage{multirow}
\usepackage{amsthm}
\usepackage{float}
\usepackage{subfig}
\usepackage{xcolor}
\usepackage{array,epsfig,fancyheadings,rotating}
\usepackage{algorithm}
\usepackage{algpseudocode}
\newtheorem{theorem}{Theorem}
\newtheorem{lemma}{Lemma}
\newtheorem{remark}{Remark}
\newcommand{\bx}{\boldsymbol{x}}
\newcommand{\bt}{\boldsymbol{\theta}}

\DeclareMathOperator*{\argmax}{arg\,max}
\DeclareMathOperator*{\argmin}{arg\,min}

\usepackage{caption}
\usepackage{float} 
\usepackage{graphicx} 

\usepackage{iftex}
\ifPDFTeX
  \usepackage[T1]{fontenc}
  \usepackage[utf8]{inputenc}
  \usepackage{textcomp} 
\else 
  \usepackage{unicode-math}
  \defaultfontfeatures{Scale=MatchLowercase}
  \defaultfontfeatures[\rmfamily]{Ligatures=TeX,Scale=1}
\fi
\usepackage{lmodern}
\ifPDFTeX\else  
\fi
\IfFileExists{upquote.sty}{\usepackage{upquote}}{}
\IfFileExists{microtype.sty}{
  \usepackage[]{microtype}
  \UseMicrotypeSet[protrusion]{basicmath} 
}{}
\makeatletter
\@ifundefined{KOMAClassName}{
  \IfFileExists{parskip.sty}{%
    \usepackage{parskip}
  }{
    \setlength{\parindent}{0pt}
    \setlength{\parskip}{6pt plus 2pt minus 1pt}}
}{
  \KOMAoptions{parskip=half}}
\makeatother
\usepackage{xcolor}
\makeatletter
\ifx\paragraph\undefined\else
  \let\oldparagraph\paragraph
  \renewcommand{\paragraph}{
    \@ifstar
      \xxxParagraphStar
      \xxxParagraphNoStar
  }
  \newcommand{\xxxParagraphStar}[1]{\oldparagraph*{#1}\mbox{}}
  \newcommand{\xxxParagraphNoStar}[1]{\oldparagraph{#1}\mbox{}}
\fi
\ifx\subparagraph\undefined\else
  \let\oldsubparagraph\subparagraph
  \renewcommand{\subparagraph}{
    \@ifstar
      \xxxSubParagraphStar
      \xxxSubParagraphNoStar
  }
  \newcommand{\xxxSubParagraphStar}[1]{\oldsubparagraph*{#1}\mbox{}}
  \newcommand{\xxxSubParagraphNoStar}[1]{\oldsubparagraph{#1}\mbox{}}
\fi
\makeatother

\usepackage{longtable,booktabs,array}
\usepackage{calc} 
\usepackage{etoolbox}
\makeatletter
\patchcmd\longtable{\par}{\if@noskipsec\mbox{}\fi\par}{}{}
\makeatother
\IfFileExists{footnotehyper.sty}{\usepackage{footnotehyper}}{\usepackage{footnote}}
\makesavenoteenv{longtable}
\usepackage{graphicx}
\makeatletter
\def\maxwidth{\ifdim\Gin@nat@width>\linewidth\linewidth\else\Gin@nat@width\fi}
\def\maxheight{\ifdim\Gin@nat@height>\textheight\textheight\else\Gin@nat@height\fi}
\makeatother
\setkeys{Gin}{width=\maxwidth,height=\maxheight,keepaspectratio}
\makeatletter
\def\fps@figure{htbp}
\makeatother

\makeatletter
\@ifpackageloaded{caption}{}{\usepackage{caption}}
\AtBeginDocument{%
\ifdefined\contentsname
  \renewcommand*\contentsname{Table of contents}
\else
  \newcommand\contentsname{Table of contents}
\fi
\ifdefined\listfigurename
  \renewcommand*\listfigurename{List of Figures}
\else
  \newcommand\listfigurename{List of Figures}
\fi
\ifdefined\listtablename
  \renewcommand*\listtablename{List of Tables}
\else
  \newcommand\listtablename{List of Tables}
\fi
\ifdefined\figurename
  \renewcommand*\figurename{Figure}
\else
  \newcommand\figurename{Figure}
\fi
\ifdefined\tablename
  \renewcommand*\tablename{Table}
\else
  \newcommand\tablename{Table}
\fi
}
\@ifpackageloaded{float}{}{\usepackage{float}}
\floatstyle{ruled}
\@ifundefined{c@chapter}{\newfloat{codelisting}{h}{lop}}{\newfloat{codelisting}{h}{lop}[chapter]}
\floatname{codelisting}{Listing}

\makeatother
\makeatletter
\@ifpackageloaded{caption}{}{\usepackage{caption}}
\@ifpackageloaded{subcaption}{}{\usepackage{subcaption}}
\makeatother

\ifLuaTeX
  \usepackage{selnolig}  
\fi
\usepackage[]{natbib}
\usepackage{bookmark}

\IfFileExists{xurl.sty}{\usepackage{xurl}}{} 
\hypersetup{
  pdftitle={Title},
  pdfauthor={Author 1; Author 2},
  pdfkeywords={3 to 6 keywords, that do not appear in the title},
  colorlinks=true,
  linkcolor={blue},
  filecolor={Maroon},
  citecolor={Blue},
  urlcolor={Blue},
  pdfcreator={LaTeX via pandoc}}

\newcommand{\anon}{1}

\begin{document}

\def\spacingset#1{\renewcommand{\baselinestretch}%
{#1}\small\normalsize} \spacingset{1.2}

\if1\anon
{
  \title{\bf Nearly Optimal Subdata Selection}
  \author{Min Yang\thanks{
     The author gratefully acknowledges \textit{NSF Grants DMS-22-10546.}}\hspace{.2cm}\\
    Department of Mathematics, Statistics, and Computer Science\\
  University of Illinois at Chicago\\
  and\\
       Wei Zheng \\
  Department of Business Analytics and Statistics \\
  University of Tennessee \\
    and \\
    John Stufken\thanks{
     The authors gratefully acknowledges \textit{NSF Grants DMS-23-04767.}}\\
    Department of Statistics \\
  George Mason University \\
    and \\
    Ming-Chung Chang \thanks{
     The author gratefully acknowledges \textit{NSTC Grants NSTC 114-2628-M-001-008-MY3 .}}\\
    Institute of Statistical Science \\
 Academia Sinica \\
 and\\
   Ting Tian \\
  School of Mathematics \\
  Sun Yat-sen University \\
and\\
  Xueqin Wang \\
  School of Management \\
  University of Science and Technology of China 
 }
 \date{}
  \maketitle
 
} \fi

\if0\anon
{
  \bigskip
  \begin{center}
    {\LARGE\bf Efficient Subdata Selection for Parameter Estimation}
\end{center}
  \medskip
} \fi

\begin{abstract}
When, in terms of the number of data points, the size of a dataset exceeds available computing resources, or when labeling is expensive, an attractive solution consists of selecting only some of the data points (subdata) for further consideration. A central question for selecting subdata of size $n$ from $N$ available data points is which $n$ points to select. While an answer to this question depends on the objective, one approach for a parametric model and a focus on parameter estimation is to select subdata that retains maximal information. Identifying such subdata is a classical NP-hard problem due to its inherent discreteness. Based on optimal approximate design theory, we develop a new methodology for information-based subdata selection, resulting in subdata that approaches the optimal solution.  To achieve this,  we develop a novel algorithm that applies to a general model, accommodates arbitrary choices of $N$ and $n$, and supports multiple optimality criteria, and we prove its convergence. 
Moreover, the new methodology facilitates an assessment of the efficiency of subdata selected by any method by obtaining tight lower and upper bounds for the efficiency. We show that the subdata obtained through the new methodology is highly efficient and outperforms all existing methods. 
\end{abstract}

\section{Introduction}

Subdata selection is a rapidly growing area of research \citep{Mirzasoleiman_2020}, driven in part by the increasing size of datasets in terms of the number of observations. It is also relevant for smaller datasets, particularly when they are unlabeled and labeling is expensive, requiring the labeling of only selected data points \citep{Tharwat_Schenck_2023}. For larger datasets, the trade-off between efficiency and computational cost must be considered, while for smaller datasets, the focus is primarily on efficiency \citep{WangYuSingh_JMLR2017}.

In both scenarios, the challenge lies in selecting $n$ observations from a dataset with $N$ observations, where $N$ is much larger than $n$, while retaining the maximum amount of information. We will refer to the original dataset of size $N$ as \textit{full data} and to the selected dataset of size $n$ as \textit{subdata}.

Several approaches have been studied for subdata selection. For example, the coreset approach aims to select subdata that can produce results close to those obtained from the full data \citep{Feldman2020CoreSets}. Another approach, aimed mainly at selecting subdata with a similar distribution as the full data, is the support points method \citep{Mak_Joseph_2018}. It minimizes the energy distance, a measure originally proposed by \cite{Szekely_Rizzo_2004} to compare distributions and test the goodness of fit. This method was further developed to work for larger datasets into the SPlit method \citep{Joseph_Vakayil_2022}. A third approach, under the assumption of a parametric model for the full data, uses the Fisher information matrix for selecting subdata \citep{Wang2019Information-BasedRegression}. Our proposed method is information-based and significantly improves upon previous work using this approach.

Three key challenges associated with subdata selection are: 
\begin{enumerate}
    \item How can we assess the quality of selected subdata?
    \item Given the full data, what subdata size $n$ yields acceptable results?
    \item Given the full data, a subdata size $n$, and a method for assessing the quality of the subdata, which $n$ data points should we select? 
\end{enumerate}

The quality of selected subdata depends partly on the study's objective. Subdata well-suited for prediction may not be ideal for estimation or exploration. We will focus on parametric statistical models with parameter estimation as the main objective. We will formulate an optimality criterion based on this objective and develop an algorithm to select subdata of size $n$ that optimizes this criterion. Previous attempts to develop such algorithms \citep{Wang2019Information-BasedRegression, Wang2021OrthogonalRegression}, have faced the challenge of the optimization problem being NP-hard, and an assessment of the performance of the selected subdata relative to the optimal solution has been unavailable. Our approach addresses this by developing a tool that shows that our algorithm selects highly efficient subdata, generally outperforming other methods.

The second focus of this paper is the third key challenge, which is addressed using a novel approach that utilizes optimal bounded approximate designs, supported by a computationally efficient algorithm.

In Section~\ref{sec:preliminary} we will provide background information, which is followed in Section~\ref{sec:method} with details about our new approach, including theoretical properties and a description of the algorithm that makes it work. In Section~\ref{sec:empirical} we demonstrate the efficiency of our method through simulation studies and make comparisons to alternative methods. We conclude with a brief discussion in Section~\ref{sec:discussion}. Proofs, more details of the algorithm, and additional simulation studies are relegated to the supplementary material.

\section{Preliminaries} \label{sec:preliminary}
In this section, we formulate the problem, connect it to concepts from optimal design of experiments, and motivate why we opt for the information-based approach. 

For the $i$th observation, $i=1,\ldots,N$, write $(y_i, \boldsymbol{x}_i)$, where $y_i$ and $\boldsymbol{x}_i$ denote the label and $p\times 1$ vector of features, respectively. Since we will only use $\boldsymbol{x}_i$ for subdata selection, the $y_i$'s need not be available. We assume that the relationship between $y_i$ and $\boldsymbol{x}_i$ can be described by a parametric linear or generalized linear model (GLM) with parameter vector $\boldsymbol{\theta}$, and that inferences about $\boldsymbol{\theta}$ are based on its Maximum Likelihood Estimator (MLE).

When the number of observations $N$ is extraordinarily large, analyzing the entire dataset can be infeasible or too expensive. Alternatively, when $y_i$'s are unavailable, such as in an active learning setup, we may need to select a subset of $\boldsymbol{x}_i$'s for which we collect the $y_i$'s if labeling is expensive. In both cases, we need to select subdata from the full data and use the subdata for inference. Using indicator variables $\delta_i$, with $\delta_i=1$ if and only if the $i$th observation belongs to the selected subdata, we want to identify $\boldsymbol{\delta} = (\delta_1, \ldots, \delta_N)$, subject to $\sum_{i=1}^N \delta_i = n$, where $n$ is the subdata size, such that inferences based on the subdata are optimal in a way to be specified. 

We will use the Fisher information matrix for $\boldsymbol{\theta}$, $I_{\boldsymbol{\theta}}$, to evaluate subdata performance. Let $I_{\boldsymbol{\theta}}(\boldsymbol{x}_i)$ represent the information matrix for $\boldsymbol{\theta}$ for the $i$th data point. Assuming independence, the information matrix for $\boldsymbol{\theta}$ using the subdata is given by
\begin{equation}\label{iboss1}
I_{\boldsymbol{\theta}}=\sum_{i=1}^N \delta_i I_{\boldsymbol{\theta}}(\boldsymbol{x}_i).
\end{equation}
 Let $\hat{\boldsymbol{\theta}}$ be the MLE of $\boldsymbol{\theta}$, which is assumed without loss of generality to be unique for the selected subdata. Since $Var(\hat{\boldsymbol{\theta}})$ is proportional to $I^{-1}_{\boldsymbol{\theta}}$, we aim for subdata that, in some sense, minimizes $I^{-1}_{\boldsymbol{\theta}}$. We borrow the approach from optimal design of experiments \citep{Kiefer1974GeneralTheory} by minimizing $\Phi(I_{\boldsymbol{\theta}})$ for a convex function $\Phi$. Two popular choices are $\Phi(I_{\boldsymbol{\theta}}) = \log(\det(I^{-1}_{\boldsymbol{\theta}}))$ and $\Phi(I_{\boldsymbol{\theta}}) = Tr(I^{-1}_{\boldsymbol{\theta}}))$, corresponding to $D$- and $A$-optimality, respectively. Thus, for some $\Phi$, we seek subdata with indicator $\boldsymbol{\delta}^{opt}$ so that
\begin{equation}\label{iboss2}
\boldsymbol{\delta}^{opt}=\arg \min_{\boldsymbol{\delta}}\Phi(I_{\boldsymbol{\theta}}).
\end{equation} 
For ease of presentation, we will focus on the $A$- and $D$-optimality criteria. However, the results also hold for other optimality criteria, such as the $\Phi_p$-optimality criteria defined in \cite{Yang_Biedermann_Tang_JASA}, provided that $p < \infty$. The proofs for the $\Phi_p$-optimality criteria are essentially the same with some additional notations.
\begin{remark}
This formulation is a general framework. For example, the subdata selection considered in \cite{WangYuSingh_JMLR2017} where $I_{\boldsymbol{\theta}}(\boldsymbol{x}_i)=\boldsymbol{x}_i\boldsymbol{x}'_i$ can be viewed the subdata selection for the linear first-order model with i.i.d. normal distribution error under $A$-optimality.   
\end{remark}

With ${N}\choose {n}$ different subdata choices, finding an exact solution to this minimization problem is computationally hard, even for modest values of $N$ and $n$. Using a result from \cite{Welch1982}, we show that a polynomial time algorithm does not exist for $D$-optimality, i.e., it is an NP-hard problem.
\begin{theorem}\label{np-hard}
$D$-optimal exact subdata selection is NP-hard. 
\end{theorem}
While Theorem~\ref{np-hard} is for $D$-optimality, we expect the result to hold for most other optimality criteria as well, including $A$-optimality. In view of Theorem~\ref{np-hard}, we propose a new polynomial time algorithm for selecting highly efficient subdata. We will prove that subdata selected by the algorithm is indeed highly efficient.

\section{The proposed method} \label{sec:method}
We consider the following continuous optimization problem, which is a  relaxation of (\ref{iboss2}):
\begin{equation}\label{psi_design}
\xi^*=\argmin_{\xi} \Phi(\xi),
\end{equation}
where $\xi=\{(\boldsymbol{x}_i,w_i), i=1,\ldots,N\}$ with $0\leq w_i\leq 1/n$ and $\sum_{i=1}^N w_i=1$, and $\Phi(\xi)= \Phi\left(\sum_{i=1}^N w_i I_{\boldsymbol{\theta}}(\boldsymbol{x}_i)\right)$. Here, $\xi$ can be interpreted as a bounded approximate design. A key tool in this framework is the general equivalence theorem, which provides necessary and sufficient conditions for optimality under criteria such as $A$- and $D$-optimality, and enables efficient identification of optimal designs.

For GLMs, $I_{\boldsymbol{\theta}}$ depends on $\boldsymbol{\theta}$. This creates a circular problem in which optimizing Equation (\ref{psi_design}) over $\xi$ is needed to find a design that optimizes inference for $\bt$, but where the function to be optimized depends on the unknown $\bt$. We will return to this in Subsection~\ref{subsec:theta}, but for simplicity we assume for now that $I_{\boldsymbol{\theta}}$ is entirely known, as is the case for linear models.

To find an optimal bounded approximate design, we will need an equivalence theorem for bounded approximate designs. To find subdata of size $n$ from an optimal bounded approximate design $\xi^*$, we will take all $x_i$'s with weight $1/n$, and supplement this set with points that have the next largest weights in $\xi^*$, until we have a total of $n$ points. This works because in an optimal bounded approximate design, there are at least $n$ points with a nonzero weight, and it works especially well if $\xi^*$ has many points with a weight close to $1/n$.

For an optimal bounded approximate design $\xi^*$, let $S^*$ denote the subdata of size $n$ obtained from $\xi^*$ as described in the previous paragraph. For arbitrary subdata $S$ of size $n$, we use the same notation $S$ for the bounded approximate design, where the indicator vector of $S$ defines an approximate design with weights $1/n$ for selected points and 0 otherwise. Furthermore, let $S_{\text{opt}}$ denote the optimal subdata, which is usually unknown. Then
\begin{equation} \label{ineq:optimal}
\Phi(\xi^*) \le \Phi(S_{\text{opt})} \le \Phi(S^*).
\end{equation}
The efficiency of the arbitrary subdata $S$ can then be defined as $\text{Eff}(S)=\frac{\Phi(S_{\text{opt}})}{\Phi(S)}$. Although this is unknown since we don't know $S_{\text{opt}}$, we can find an upper and lower bound for this efficiency by using Equation.~(\ref{ineq:optimal}):
\begin{equation}\label{eff_bounds}
\frac{\Phi(\xi^*)}{\Phi(S)}\leq \text{Eff}(S)\leq \frac{\Phi(S^*)}{\Phi(S)}=\left. \frac{\Phi(\xi^*)}{\Phi(S)} \middle/ \frac{\Phi(\xi^*)}{\Phi(S^*)} \right..
\end{equation}
Since $\frac{\Phi(\xi^*)}{\Phi(S^*)}$ is usually very close to 1 \citep{Pukelsheim_Rieder_biometrika}, it follows from Equation~(\ref{eff_bounds}) that $\text{Eff}(S)$ is approximately equal to its lower bound $\Phi(\xi^*)/\Phi(S)$. In our empirical studies in Section~\ref{sec:empirical}, we observe that $S^*$ is indeed highly efficient with an efficiency lower bound of at least 99.9\% for large $N$ and $n$. As a result, the upper and lower bounds for $\text{Eff}(S)$ will be almost the same. 

We point out that when computing the efficiency for $D$-optimality, we use the commonly accepted form $\Phi(I_{\boldsymbol{\theta}}) = \det(I^{-1}_{\boldsymbol{\theta}})^{1/p_1}$, where $p_1$ is the dimension of $I_{\boldsymbol{\theta}}$. The form $\Phi(I_{\boldsymbol{\theta}}) = \log(\det(I_{\bt}^{-1}))$ given in Section~\ref{sec:preliminary} is however used for the Equivalence Theorem in Subsection~\ref{subsec:ET}.

Since all of the above hinges on finding the optimal bounded approximate design $\xi^*$, we now turn to the theorem and algorithm that will facilitate this.

\subsection{The Equivalence Theorem for bounded approximate designs} \label{subsec:ET} 
Before presenting the first theorem, we define the class of optimality criteria considered in this paper. The optimality criterion $\Phi$ must be convex and linearly differentiable. The latter means that for approximate designs $\xi$ and $\eta = \{(\bx_i,\lambda_i),i=1,...,N\}$ it holds that 
$$
F_\Phi(\xi,\eta) = \sum_{i=1}^N \lambda_i F_\Phi(\xi;\bx_i),
$$
where 
$$
F_\Phi(\xi,\eta) = \lim_{\alpha \downarrow 0} \frac{\Phi((1-\alpha)\xi+\alpha\eta)-\Phi(\xi)}{\alpha}
$$
is the directional derivative of $\Phi$ at $\xi$ in the direction of $\eta$. Also, $F_\Phi(\xi;\bx_i)$ is the directional derivative of $\Phi$ at $\xi$ in the direction of the Dirac measure $\delta_{\bx_i}$. An additional property that we will use is that for designs $\xi_1$, $\xi_2$ and $\xi_3$, for all positive values $\gamma$ for which $\xi = \xi_1 + \gamma(\xi_2 - \xi_3)$ is again a design, it holds that $\Phi(\xi)$ is infinitely differentiable in $\gamma$. A class of optimality criteria that meets these conditions is the class of $\Phi_p$-optimality criteria for finite $p$ (cf. \cite{Yang_Biedermann_Tang_JASA}), which includes $D$- and $A$-optimality. The latter two are the main focus of this paper. In the statements of the theorems, we will use $\Lambda$ for the class of optimality criteria that satisfy the above conditions.

The following theorem is adapted from \cite{Sahm2001ADesigns}. A major difference is that the design space $\mathcal{X}$ is discrete in our case and continuous in \cite{Sahm2001ADesigns}. In terms of weight constraints, the formulation of the theorem, which uses the lower bound $\nu_i$ and the upper bound $\mu_i$ for the weight $w_i$ is more general than we will need: For our application $\nu_i=0$ and $\mu_i=1/n$ for all $i=1,\ldots,N$.
\begin{theorem}\label{eq_thm} 
(Equivalence Theorem for bounded approximate designs) Let $\{\nu_1,\ldots,\nu_N\}$ and $\{\mu_1,\ldots,\mu_N\}$ be two sets of numbers so that $0\le \nu_i < \mu_i$, $\sum_{i=1}^N \nu_i \le 1$, and $\sum_{i=1}^N \mu_i \ge 1$.
Let $\Xi$ be the set of all bounded approximate designs $\{(\boldsymbol{x}_i,w_i), i=1,\ldots,N\}$ with $ \nu_i\leq w_i\leq \mu_i$ and $\sum_{i=1}^N w_i=1$, and let $\Phi \in \Lambda$. The following two statements are equivalent:
\begin{itemize}
    \item[(I)] $\xi^*=\{(\boldsymbol{x}_i,w_i), i=1,\ldots, N\}$ is $\Phi$-optimal in $\Xi$.
    \item[(II)] There is a partition of the design space, $\mathcal{X} = \mathcal{X}_1 \cup \mathcal{X}_2 \cup \mathcal{X}_3$, and a number $s$ such that
    \begin{itemize}
        \item[(a)] $w_i=\nu_i$ for $\boldsymbol{x}_i\in\mathcal{X}_1$ and $w_i=\mu_i$ for $\boldsymbol{x}_i\in\mathcal{X}_3$;
    \item[(b)] If $\mathcal{X}_2$ is empty, then  $\max\limits_{\boldsymbol{x}_i\in\mathcal{X}_3}F_\Phi(\xi^*;\boldsymbol{x}_i) \leq \min\limits_{\boldsymbol{x}_i\in\mathcal{X}_1}F_\Phi(\xi^*;\boldsymbol{x}_i)$; and
    \item[(c)] If $\mathcal{X}_2$ is nonempty, then there is a number $s$ so that for every $\boldsymbol{x}_i\in\mathcal{X}_2$
    \begin{itemize}
        \item[(i)] $\nu_i< w_i< \mu_i$;
        \item[(ii)] $F_\Phi(\xi^*;\boldsymbol{x}_i)=s$; and 
        \item[(iii)] $\max\limits_{\boldsymbol{x}_i\in\mathcal{X}_3}F_\Phi(\xi^*;\boldsymbol{x}_i) \leq s \leq \min\limits_{\boldsymbol{x}_i\in\mathcal{X}_1}F_\Phi(\xi^*;\boldsymbol{x}_i)$.
    \end{itemize} 

    \end{itemize}
\end{itemize}

\end{theorem}
A proof of Theorem~\ref{eq_thm}, inspired by the proof in \cite{Sahm2001ADesigns}, is given in the supplementary material. It can be shown that if (I) in Theorem~\ref{eq_thm} holds, then the value of $s$ in (II)(c) is given by
    \begin{eqnarray}
   s=\frac{-\sum_{\boldsymbol{x}_i\in\mathcal{X}_1}F_\Phi(\xi^*;\boldsymbol{x}_i)\nu_i-\sum_{\boldsymbol{x}_i\in\mathcal{X}_3}F_\Phi(\xi^*;\boldsymbol{x}_i)\mu_i}{1-(\sum_{\boldsymbol{x}_i\in\mathcal{X}_1}\nu_i+\sum_{\boldsymbol{x}_i\in\mathcal{X}_3}\mu_i)}.
\end{eqnarray} 
As already noted, for the subdata selection problem, we have $\nu_i=0$ and $\mu_i=1/n$, $i=1,\ldots,N$.
 
\subsection{Identifying an optimal bounded approximate design: The main algorithm}
As observed just prior to Subsection~\ref{subsec:ET}, the proposed methodology depends on identifying a $\Phi$-optimal bounded approximate design $\xi^*$. Theorem~\ref{eq_thm}, with $\nu_i=0$ and $\mu_i=1/n$, $i=1,\ldots, N$, will guide us in finding such a design via an iterative algorithm. 

For an index $t$, during the algorithm, we will have a partition $\mathcal{P}^{(t)}=(\mathcal{X}^{(t)}_1, \mathcal{X}^{(t)}_2, \mathcal{X}^{(t)}_3)$ of the design space $\mathcal{X}$ with weights $w_i^{(t)}$. As in Theorem~\ref{eq_thm}, $\mathcal{X}^{(t)}_1$ and $\mathcal{X}^{(t)}_3$ contain the points in $\mathcal{X}$ with weights 0 and $1/n$, respectively, for a bounded approximate design $\xi_{\mathcal{P}^{(t)}}$ at that stage of the algorithm, while $\mathcal{X}^{(t)}_2$ contains the points, if any, with optimally selected weights that are strictly between these two extremes. The following lemma provides an interesting property for these types of designs.

\begin{lemma}\label{lemma1}
For $\Phi \in \Lambda$, consider a design $\xi_\mathcal{P}$ corresponding to a partition $ \mathcal{P} = (\mathcal{X}_1, \mathcal{X}_2, \mathcal{X}_3)$ in which the weights for points in $\mathcal{X}_2$, if any, have been optimized and are strictly between 0 and $1/n$. If $\mathcal{X}_2\neq \emptyset$, then $F_{\Phi}(\xi_{\mathcal{P}};\bx_i)$ is constant for $\bx_i \in \mathcal{X}_2$.
\end{lemma}
Note that part (c) of Theorem~\ref{eq_thm} makes this statement already for a $\Phi$-optimal design $\xi^*$.

We present pseudocode for finding a $\Phi$-optimal design in Algorithm~\ref{alg:main}. As shown in Algorithm~\ref{alg:main}, we first search for a reasonable starting design based on subdata, followed by an iteration step to find an optimal bounded approximate design. During the iteration, optimal weights must be computed for points in $\mathcal{X}_2^{(t)}$ (see line~\ref{alg:optw}). We will present an efficient algorithm for this step in Subsection~\ref{implementation}. It is also worth pointing out that the proofs will show that the iteration procedure of Algorithm~\ref{alg:main} results in a better design if conditions (II)(b) or (II)(c) of Theorem~\ref{eq_thm} are not satisfied. That is, in that case $\Phi(\mathcal{P}^{(t+1)}) < \Phi(\mathcal{P}^{(t)})$.

\begin{algorithm}[H]
\caption{Main Algorithm}\label{alg:main}
\begin{small}
\begin{algorithmic}[1]
\State \textbf{Input}: $\bx_i$, $i=1,\ldots,N$; subdata size $n$; optimality criterion $\Phi \in \Lambda$; information matrix $I_{\bt}(\bx_i)$ of dimension $p_1$
\State \textbf{Output}: Optimal bounded approximate design $\xi^*$

\Procedure{Preparation}{}
  \State Using IBOSS \citep{Wang2019Information-BasedRegression}, select subdata $S$ of size $n$ and let $\xi_{S}$ be the approximate design as in Eq.~(\ref{psi_design})\label{alg:startS} 
  \State Compute $F_{\Phi}(\xi_{S},\bx_i)$, $i=1,\ldots,N$ \label{alg:comp}
  \State Form $E_1$ consisting of $\lfloor n/p_1\rfloor$ points $\bx_i \not \in S$ with the smallest values for $F_{\Phi}(\xi_{S},\bx_i)$
  \State Form $E_3$ consisting of $\lfloor n/p_1\rfloor$ points $\bx_i \in S$ with the largest values for $F_{\Phi}(\xi_{S},\bx_i)$
  \State Assign $S = (S \backslash E_3) \cup E_1$\label{alg:assign}
  \State Repeat lines~\ref{alg:comp}--\ref{alg:assign} $p_1$ times\label{alg:finalS}
  \State Obtain $\bx_{1*}=\argmin_{\bx_i\not\in S}F_{\Phi}(\xi_{S},\bx_i)$ and $\bx_{3*}=\argmax_{\bx_i\in S}F_{\Phi}(\xi_{S},\bx_i)$\label{alg:onepoint}
  \State Assign $S = (S \backslash \{\bx_{3*}\}) \cup \{\bx_{1*}\}$\label{alg:newS}
  \State Repeat lines~\ref{alg:onepoint}--\ref{alg:newS} $n$ times
  \State \label{desP0} Let $\xi_{\mathcal{P}^{(0)}}$ be the approximate design corresponding to the partition $\mathcal{P}^{(0)} = (\mathcal{X}^{(0)}_1, \mathcal{X}^{(0)}_2, \mathcal{X}^{(0)}_3)$, where $\mathcal{X}^{(0)}_1 = \mathcal{X} \backslash S$, $\mathcal{X}^{(0)}_2 = \emptyset$, and $\mathcal{X}^{(0)}_3 = S$
\EndProcedure

\Procedure{Iteration}{}
  \State Set $t=0$
  \State If $\mathcal{X}^{(t)}_2 \not = \emptyset$, compute the optimal weights between $0$ and $1/n$ for points $\bx_i \in \mathcal{X}^{(t)}_2$ and redefine the partition $\mathcal{P}^{(t)}$ if some weights become $0$ (move those points to $\mathcal{X}^{(t)}_1$) or $1/n$ (move those points to $\mathcal{X}^{(t)}_3$)\label{alg:optw}
  \State If conditions (II)(b) and (II)(c) in Theorem~\ref{eq_thm} are satisfied for $\xi_{\mathcal{P}^{(t)}}$, set $\xi^* = \xi_{\mathcal{P}^{(t)}}$ and go to line~\ref{alg:stop}\label{alg:condIIb}
  \State Else, obtain $\bx_{1*}=\argmin_{\bx_i\in \mathcal{X}^{(t)}_1}F_{\Phi}(\xi_{\mathcal{P}^{(t)}},\bx_i)$ and $\bx_{3*}=\argmax_{\bx_i\in \mathcal{X}^{(t)}_3}F_{\Phi}(\xi_{\mathcal{P}^{(t)}},\bx_i)$
  \State Define $\mathcal{X}^{(t+1)}_1=\mathcal{X}^{(t)}_1\setminus \{\bx_{1*}\}$,\;
         $\mathcal{X}^{(t+1)}_2=\mathcal{X}^{(t)}_2\cup \{\bx_{1*},\bx_{3*}\}$,\;
         $\mathcal{X}^{(t+1)}_3=\mathcal{X}^{(t)}_3\setminus \{\bx_{3*}\}$\label{alg:newP}
  \State Set $t = t+1$\label{alg:increase}
  \State Repeat lines~\ref{alg:optw}--\ref{alg:increase} until $\Phi(\xi_\mathcal{P}^{(t)})$ converges\label{alg:conv}
\EndProcedure\label{alg:stop}

\end{algorithmic}
\end{small}
\end{algorithm}

An important property of any algorithm is that it converges. Our algorithm is designed for a general model, arbitrary choices of $N$ and $n$, and multiple optimality criteria. To establish convergence in such a general setting, we introduce an assumption.  Notice that this assumption is only needed in Theorems \ref{convg} and \ref{speed}.

 Let $\Xi$ again denote the set of all bounded approximate designs $\xi=\{(\bx_i,w_i), i=1,\ldots,N\}$, but now for $0\leq w_i\leq 1/n$ and $\sum_{i=1}^Nw_i=1$. Let  
 $\xi \in \Xi$ be a design in the iteration stage of Algorithm~\ref{alg:main},  $\bx_{+}$ and $\bx_-$ be two points with weights $w_+$ and $w_-$ in $\xi$, respectively, so that $w_+ < 1/n$ and $w_- > 0$. For small enough $\gamma>0$, we can then define a new design $\xi[\bx_+, \bx_-,\gamma] \in \Xi$ such that
		\begin{itemize}
			\item[(1)] $w_{i}(\xi[\bx_+, \bx_-,\gamma])=w_i$ for $\bx_i\in \mathcal{X}\setminus \{\bx_+,\bx_-\}$;
			\item[(2)] $w_{i}(\xi[\bx_+, \bx_-,\gamma])=w_{+}+\gamma$ for $x_i =\bx_+$; and
			\item[(3)] $w_{i}(\xi[\bx_+, \bx_-,\gamma])=w_{-}-\gamma$ for $x_i =\bx_-$.
		\end{itemize}

We will make the following assumption.

\textbf{Assumption 1.} For all designs $\xi$ generated during the iteration procedure, it holds for a constant $C>0$ that $\Phi(\xi[\bx_+, \bx_-,\gamma])< C\Phi(\xi)$ for all $\gamma$ with $\xi[\bx_+, \bx_-,\gamma] \in \Xi$. 

\begin{remark} Although this assumption may seem restrictive, it holds under mild regularity conditions. The construction of $\xi[\bx_+, \bx_-,\gamma] $ simply decreases the weight of one point and increases the weight of another (the resulting weights are still bounded between 0 and $1/n$), while leaving all other weights unchanged. 
Since the weight for any point is at most $1/n$, this adjustment will not significantly increase the value $\Phi$, except in extreme cases (for example, when the size of the sub-data $n$ is very small, such as being equal to the dimension $p$ of the feature space). However, in practice, the size $n$ is typically much larger than $p$. Assumption 1 can also be easily verified empirically. In our examples, $\Phi(\xi[\bx_+, \bx_-,\gamma])<1.01 \Phi(\xi)$ holds for all cases.
 \end{remark}
 
The next theorem shows that Algorithm~\ref{alg:main} performs as desired.

\begin{theorem}\label{convg} 
Let $\Phi \in \Lambda$. With $\xi_{\mathcal{P}_0}$ as the bounded approximate design in line~\ref{psi_design} of Algorithm~\ref{alg:main}, if the information matrix for $\xi_{\mathcal{P}^{(0)}}$ is non-singular and Assumption 1 holds, then the $\Phi$-values for the sequence of designs $\xi_{\mathcal{P}^{(t)}}$ converge to the $\Phi$-value of an optimal bounded approximate design as $t$ increases. That is, 
 \begin{equation}
     \lim_{t\to\infty}\Phi(\xi_{\mathcal{P}^{(t)}}) = \min_{\xi\in \Xi}\Phi(\xi).
 \end{equation} 
\end{theorem}
The result in Theorem~\ref{convg} is desirable, but provides no information about the speed of convergence. Furthermore, strict verification of conditions (II)(b) and (II)(c) in Theorem~\ref{eq_thm} is infeasible due to floating-point errors in computing. Instead, we define a design $\xi$ to be an $\epsilon$-approximation of the optimal design $\xi^*$ if $\Phi(\xi)-\Phi(\xi^*) \le \epsilon$. 
In practice, we aim for an $\epsilon$-approximation of an optimal design for small values of $\epsilon$. We first give a sufficient condition for a design to be an $\epsilon$-approximation of a $\Phi$-optimal design.

\begin{theorem}\label{eps_equiv}
Let $\Xi$ be the set of all bounded approximate designs $\{(\boldsymbol{x}_i,w_i), i=1,\ldots,N\}$ with $ 0\leq w_i\leq 1/n$ and $\sum_{i=1}^N w_i=1$, and let $\Phi \in \Lambda$. For $\epsilon>0$, a design $\xi_{\epsilon}=\{(\boldsymbol{x}_i,w_i), i=1,\ldots, N\}$ is an $\epsilon$-approximation of a $\Phi$-optimal design if there is a partition of the design space, $\mathcal{X} = \mathcal{X}_1 \cup \mathcal{X}_2 \cup \mathcal{X}_3$ and 
    \begin{itemize}
        \item[(a)] $w_i=0$ for $\boldsymbol{x}_i\in\mathcal{X}_1$ and $w_i=1/n$ for $\boldsymbol{x}_i\in\mathcal{X}_3$;
    \item[(b)] if $\mathcal{X}_2 = \emptyset$, then $\max\limits_{\boldsymbol{x}_i\in\mathcal{X}_3}F_\Phi(\xi_{\epsilon};\boldsymbol{x}_i)\leq \min\limits_{\boldsymbol{x}_i\in\mathcal{X}_1}F_\Phi(\xi_{\epsilon};\boldsymbol{x}_i) + \epsilon$
    \item[(c)] if $\mathcal{X}_2 \not = \emptyset$, then there is a number $s$ so that for every $\bx_i \in \mathcal{X}_2$
    \begin{itemize}
        \item[(i)] $0 < w_i < 1/n$;
        \item[(ii)] $s-\frac{\epsilon}{2} < F_{\Phi}(\xi_{\epsilon}; \bx_i) < s + \frac{\epsilon}{2}$;
        \item[(iii)] $\max\limits_{\boldsymbol{x}_i\in\mathcal{X}_3}F_\Phi(\xi_{\epsilon};\boldsymbol{x}_i)\leq s + \frac{\epsilon}{2} \text{ and } s-\frac{\epsilon}{2} \leq \min\limits_{\boldsymbol{x}_i\in\mathcal{X}_1}F_\Phi(\xi_{\epsilon};\boldsymbol{x}_i)$.
    \end{itemize}
    \end{itemize}
\end{theorem}

The following theorem provides a result for the computational complexity of Algorithm~\ref{alg:main} for finding a design $\xi$ that is an $\epsilon$-approximation of an optimal design $\xi^*$ under $\Phi$.

\begin{theorem}\label{speed}
Let $\Phi \in \Lambda$  and Assumption 1 holds. For $\epsilon>0$, if conditions (II)(b) and II(c) of Theorem~\ref{eq_thm} in line~\ref{alg:condIIb} of Algorithm~\ref{alg:main} are replaced by (b) and (c) of Theorem \ref{eps_equiv},
then Algorithm~\ref{alg:main} is a polynomial time algorithm for finding an $\epsilon$-approximation $\xi_{\mathcal{P}^{(t)}}$ of $\xi^*$. 
\end{theorem} 

\begin{remark} \label{rem:1}
    Rather than using Algorithm~\ref{alg:main} as presented, we make the substitutions suggested by Theorem~\ref{speed} and stop the algorithm when an $\epsilon$-approximation of $\xi^*$ has been found. The proof of Theorem~\ref{speed} also guarantees that, for a fixed $\epsilon$, in the worst case, the computational complexity of Algorithm~\ref{alg:main} is $O(nN^2p^2)$. Our numerical examples require relatively few iterations for the identification of an $\epsilon$-approximation of $\xi^*$.
\end{remark}


\begin{remark} \label{rem:2}
    Since $\Phi(\xi)$ is not scale-invariant, it may be unclear how to select a meaningful value for $\epsilon$ in Theorem~\ref{speed}. Alternatively and easier to interpret, we could consider the relative distance in $\Phi$-values for a design $\xi$ and the $\Phi$-optimal design $\xi^*$, namely $\left( \Phi(\xi) - \Phi(\xi^*) \right) / \Phi(\xi^*)$. We may think of this ratio as the reciprocal of the relative $\Phi$-efficiency of $\xi$ minus 1. To demonstrate how Theorem~\ref{speed} can help to find a design that achieves a desired relative distance to a $\Phi$-optimal design, suppose that we wish to find a design $\xi$ with $\left( \Phi(\xi) - \Phi(\xi^*) \right) / \Phi(\xi^*) \le \alpha$ for some specified $\alpha > 0$. At any stage of the algorithm, if $\xi_0$ is a candidate design, use $\epsilon = \alpha \Phi(\xi_0)/ (\alpha + 1)$ in Theorem~\ref{speed}. With this $\epsilon$, we see that if $\Phi(\xi_0) - \Phi(\xi^*) \le \epsilon$, then
    $$
    \frac{\Phi(\xi_0) - \Phi(\xi^*)}{\Phi(\xi^*)} \le \frac{\epsilon}{\Phi(\xi_0) - \epsilon} = \alpha.
    $$
    Thus, if $\xi_0$ satisfies conditions (b) and (c) of Theorem~\ref{eps_equiv} for this value of $\epsilon$, then it satisfies $\left( \Phi(\xi) - \Phi(\xi^*) \right) / \Phi(\xi^*) \le \alpha$. If it does not satisfy the conditions, we improve $\xi_0$ by continuing to run the algorithm, get a design $\xi_1$, recalculate $\epsilon$, and check the conditions of Theorem~\ref{speed} again, but now with this new value of $\epsilon$. We continue until we find a design and the corresponding value of $\epsilon$ for which the conditions in Theorem~\ref{speed} hold.  
\end{remark} 

\begin{remark} \label{rem:3}
    We will refer to the subdata obtained at the end of line~\ref{alg:finalS} of Algorithm~\ref{alg:main} as IBOSS+ subdata. We will see in Section~\ref{sec:empirical} that the IBOSS+ subdata can already be highly efficient.
\end{remark}

\begin{remark} 
Theorems~\ref{convg} and \ref{speed} provide theoretical guarantees for the main algorithm, but several implementation choices affecting computational efficiency remain. In the Appendix, we address these issues, including the selection of the initial subdata $S$ in line~\ref{alg:startS} and the computation of optimal weights in line~\ref{alg:optw} of Algorithm~\ref{alg:main}.
\end{remark}

\subsection{Subdata selection based on an optimal bounded approximate design}
Once we have found an optimal bounded approximate design $\xi^*$, we need to convert it to subdata $S$ of size $n$. Equivalently, we need to convert it to a bounded approximate design $\xi_S$ with weights that are only 0 and $1/n$ or to an exact design with $n$ distinct points. Although there are algorithms to convert an approximate design to an exact design, such as the rounding algorithm in \cite{Pukelsheim_Rieder_biometrika} and the branch-and-bound algorithm in \cite{Ahipasaoglu_StatComp}, our problem allows for a simpler highly efficient solution. Typically, multiple points reach the upper bound of $1/n$ in $\xi^*$. We include all of these points in the subdata, and add the points with the largest weights until we reach the desired size of $n$.  Thus, the selected subdata consist of the $n$ points with the largest weights in $\xi^*$. Observe that $\xi^*$ has at least $n$ points with a positive weight due to the constraint $0\leq w_i\leq 1/n$.
 
 In the next section, we will see that the lower bound of the efficiency of subdata selected by this method is generally more than 99.99\%.

\section{Simulations and comparisons}\label{sec:empirical} 
The previous sections explained the methodology and derived its theoretical properties. This section will demonstrate its excellent performance. Where appropriate, we will compare our method to the following alternatives: SRS: Simple Random Sampling; ODB: Optimal design
based sampling \citep{DeldossiTommasi2022}; LEV: Leverage sampling \citep{Ma2015ALeveraging};  IBOSS: Information-based optimal subdata selection \citep{Wang2019Information-BasedRegression}; OSS: Orthogonal subsampling \citep{Wang2021OrthogonalRegression}; SEQ: Sequential sampling \citep{Pronzato_Wang_JSPI}; and FO: Using the first-order algorithm to find optimal bounded approximate designs \citep{Ahipasaoglu_StatComp}.  Most of these methods are developed under parameter estimation frameworks and are primarily designed for settings in which all model parameters are of interest.

We refer to our method as IBOSS OBD, which enhances IBOSS using optimal bounded designs. We also propose two related algorithms: IBOSS+, obtained by applying lines~\ref{alg:comp}–\ref{alg:finalS} of Algorithm~\ref{alg:main} after IBOSS, and IBOSS++, which further applies lines~\ref{alg:onepoint}–\ref{desP0}. In contrast to the methods above, the proposed framework is not restricted to targeting all parameters simultaneously and can be applied more flexibly when interest lies in either the full parameter vector or specific components.

IBOSS and OSS are specifically designed for first-order models, and will only be used for those situations. For other linear models or generalized linear models, all other models will be considered. We will also look at simulated data from a clusterwise linear model, where we can use SRS, SEQ, IBOSS+, IBOSS++ and IBOSS OBD. Since LEV and FO require a specific form of the information matrix, it is not clear how to use these methods for that model. As in Section~\ref{sec:method}, we use $S^*$ to denote subdata obtained by one of the methods, adding one of the above abbreviations as a subscript to identify the method.

The SEQ method produces a random sample size that is intended to be close to a prespecified target value 
$n$. Since all other methods can be applied for any given sample size, a natural strategy for fair comparison is to first apply SEQ and then use its realized sample size for the remaining methods. However, in practice this approach does not always perform well, as the sample size produced by SEQ can deviate substantially from the target. In such cases, SEQ is excluded from the comparison.  Additional results illustrating this issue are provided in Appendices \ref{add_simu_2} and \ref{add_simu_3}.

For the FO algorithm, the description in \cite{Ahipasaoglu_StatComp} does not provide sufficient guidance on selecting an initial design or determining the maximum number of iterations.  We investigate these issues in Appendix \ref{about_FO}. We also considered the PGD method (projected gradient descent–based sampling) \citep{WangYuSingh_JMLR2017}. However, due to its scalability limitations (see Appendix \ref{about_PGD}), we exclude PGD from our simulation comparisons.

We consider three models: a linear first-order regression model, a full second-order logistic regression model, and a clusterwise linear model. The feature vectors, $\bx_i$, $i=1,\ldots,N$, are generated from the multivariate normal distribution $N_p(\boldsymbol{1}_p, \Sigma_p)$, where the diagonal and off-diagonal elements of $\Sigma_p$ are 1 and 0.5, respectively. For each scenario, results are based on averaging over 100 simulated data sets (that is, 100 repetitions), each time generated according to the scenario under consideration. Because of space limitation, we only report the result of Secenario (i) in the main text.The results of Scenarios (ii) and (iii) are reported in Appendix \ref{add_simu_2} and \ref{add_simu_3}. 

All coding was done in Julia, and computations were performed on a Dell Laptop (2.3 GHz and 16 GB RAM). The $\epsilon$ in Theorem \ref{eps_equiv} is set to be $10^{-5}$ for all examples. For each configuration, the reported results are based on 100 replications. While we will report the computing times (in seconds) of the various methods, it should be noted that there could be tricks for some methods that we are unaware of. In particular, the OSS method, which is only a candidate for Scenario (i), shows larger computing times than we expected. It is possible that more efficient coding could result in shorter computing times.

{\bf Scenario (i): A linear first-order model for $\boldsymbol{p=10}$.} 
Under the $D$-optimality criterion, the efficiency ratio $\Phi(\xi^*)/\Phi(S_{IBOSS\ BOD}^*)$ is 99.999\%. Since this represents the lower bound for the efficiency of $S_{IBOSS\ BOD}^*$ (as given in Equation~\ref{eff_bounds}), we treat it as the true efficiency. Consequently, we report this ratio as the efficiency metric for the subdata obtained by other methods as well.

We first consider $D$-optimality when all parameters are of interest. From Table~\ref{linear_main_D}, LEV shows a modest improvement over SRS, although both methods remain considerably less efficient than ODB.  IBOSS outperforms ODB, and OSS further improves upon IBOSS. SEQ and IBOSS+ also demonstrate high efficiency, whereas IBOSS++, IBOSS OBD, and FO achieve near-optimal performance. In terms of computation time, SRS is the fastest, followed by IBOSS and LEV. IBOSS+ requires slightly more time than LEV, followed by SEQ, while IBOSS++ is slower. Among the most efficient methods, IBOSS OBD is substantially faster than FO, although both achieve a 100\% convergence rate.  In contrast, ODB exhibits a less favorable performance profile, with lower efficiency than several competing methods and considerably higher computational cost. 
The average subdata size is close to the target of $n = 1000$. 

\begin{table}
\caption{Scenario (i): $D$-optimality for all parameters ($N=100000$ and $n=1000$)}
\centering
\resizebox{\linewidth}{!}{%
\begin{tabular}{lcccccccccc}
\hline
Algorithm & SRS &ODB & LEV & IBOSS & OSS & SEQ & FO & IBOSS+ & IBOSS++ & IBOSS OBD \\
\hline
Mean Eff (\%) & 41.82 & 61.84 & 42.06 & 72.33 & 83.69 & 96.49 & 100.00 & 99.67 & 100.00 & 100.00 \\
Std Eff (\%)  & 0.60 &0.69 & 0.57  & 0.54  & 2.42  & 0.15   & 0.00 & 0.04  & 0.00  & 0.00  \\
\hline
Mean Time (s) & 0.0001 &1.9365 & 0.0465 & 0.0340 & 2.6261 & 0.2428 & 5.0611 & 0.1455 & 0.3385 & 0.6239 \\
Std Time (s)  & 0.0000 &0.6099 &0.0087 & 0.0096 & 0.3287 & 0.0518 & 0.6286 & 0.01833 & 0.0405 & 0.1604 \\
\hline
\multicolumn{10}{l}{\footnotesize Mean(size) = 1021.59, Std(size) = 9.94; convergence rates of IBOSS OBD and FO are both 1.00.} \\
\hline
\end{tabular}%
}
\label{linear_main_D}
\end{table}

 We next consider A-optimality when only the first five parameters are of interest. The results in Table~\ref{linear_main_A} closely resemble those in Table~\ref{linear_main_D}, but with two notable differences. First, the proposed methods continue to achieve high efficiency, whereas the competing methods are substantially less efficient than in the D-optimality setting, particularly SEQ and FO, which are no longer effective. This is expected, as most competing methods are designed for settings in which all model parameters are of interest. In contrast, the proposed approach remains effective when attention is restricted to a subset of parameters. Second, the average computation time for FO increases substantially to about 46 seconds, compared with approximately 5 seconds under D-optimality. In contrast, the computation time for IBOSS OBD remains essentially unchanged, and IBOSS OBD converges in all repetitions.

Since ODB does not provide a competitive balance between statistical efficiency and scalability, it is excluded from further consideration in Scenarios (ii) and (iii) in Appendices  (\ref{add_simu_2}) and \ref{add_simu_3}

\begin{table}
\centering
\caption{Scenario (i): $A$-optimality for the first five slope parameters ($N=100000$ and $n=1000$)}
\resizebox{\linewidth}{!}{%
\begin{tabular}{lcccccccccc}
\hline
Algorithm & SRS  &ODB & LEV & IBOSS & OSS & SEQ & FO & IBOSS+ & IBOSS++ & IBOSS OBD  \\
\hline
Mean Eff (\%) & 27.82 &57.42 & 28.40 & 43.57 & 53.34 & 68.83 & 72.77 & 98.90 & 99.98 & 100.00 \\
Std Eff (\%)  & 0.67  & 1.11& 0.67 &0.71 & 2.19  & 0.67  & 0.43  & 0.15  & 0.01  & 0.00  \\
\hline
Mean Time (s) & 0.0001 &  1.0670 & 0.0463 & 0.0359 & 2.3201 & 0.3707 & 45.9827 & 0.1618&0.3317 & 0.6185 \\
Std Time (s)  & 0.0000 &  0.2100 & 0.0116 & 0.0173 & 0.3350 & 0.0481 & 4.8112 & 0.0245 & 0.0589 &0.1143 \\
\hline
\multicolumn{10}{l}{\footnotesize Mean(size) = 928.36, Std(size) = 66.92; convergence rates: IBOSS OBD = 1.00.} \\
\hline
\end{tabular}%
}
\label{linear_main_A}
\end{table}

\section{Discussion} \label{sec:discussion}
Selecting subdata is essential for analyzing large datasets, especially when computational resources are limited. It is also key for choosing training data in machine learning when labels are unavailable or costly to obtain. However, finding optimal subdata is an NP-hard problem, meaning that practical solutions can only aim for efficiency rather than perfection. While many algorithms exist for efficient subdata selection, their actual effectiveness is not well understood. This paper introduces the IBOSS OBD algorithm, which not only produces highly efficient subdata but also provides a way to measure the efficiency of subdata from other methods. The algorithm searches for bounded optimal approximate designs and is theoretically proven to converge to the true optimal design. Simulations show that it works well across various models, including linear, generalized linear, and more complex ones like clusterwise linear regression.

The computing time of IBOSS OBD increases for more complex models simply because optimization for high-dimensional matrices is a challenging problem. However, empirical evidence suggests that the IBOSS+ method, which provides the starting point for the IBOSS OBD method, provides highly efficient subdata for a fraction of the IBOSS OBD computing time. Thus, IBOSS+ is a very attractive method that yields highly efficient subdata for a wide variety of different models and that generally substantially outperforms methods like SRS,  PGD,  LEV, IBOSS, OSS,  ODB,  SEQ and FO, which appear on top of that not always to be applicable. Note that we can only make the claim about the high efficiency of IBOSS+ because of the development of our IBOSS OBD methodology. 

Most existing approaches are tailored to settings where the full parameter vector is of interest. In contrast, the proposed framework accommodates both full and partial parameter vector. When all parameters are of interest,
the most formidable competitors to IBOSS+, IBOSS++, and IBOSS OBD are SEQ and FO. However, for neither is it clear how to apply them to the clusterwise linear regression model in Scenario (iii). Moreover, even when it converges, FO is slower than IBOSS OBD. On top of that, we found that FO takes a long time to converge for $A$-optimality. SEQ also performed poorly for $A$-optimality in Scenario (ii), where the realized sample sizes are often substantially below the target, leading to reduced statistical efficiency.

While this paper offers a very efficient method for subdata selection, there are various issues that go beyond the scope of this paper and that will require further study. First, for small datasets (which could be encountered in some of the machine learning applications), the lower and upper bounds in Equation~(\ref{eff_bounds}) for the efficiency of the subdata may be further apart than in the empirical studies conducted here. This is simply because for smaller datasets, the ratio $\Phi(\xi^*)/\Phi(S^*)$ in Equation (\ref{eff_bounds}) may not be as close to 1. Hence, in such a case there will be more uncertainty in the efficiency of subdata $S$. Second, while the IBOSS OBD method converged almost all of the time for the tuning parameters that we selected, occasionally it did not (cf. Table~\ref{table:logitD}, where convergence did not occur in 1\% of the repetitions). But even in those cases, the $\Phi$-values for the bounded designs that had been obtained were essentially equal to those of the bounded optimal design. We are not aware of any situations where this is not the case, but may simply not have encountered situations where a different value for the tuning parameters should be selected. Third, as in the case of Scenario (iii), more complex models may not have an explicit information matrix. Clusterwise linear regression models are a special and relatively simple case of mixture-of-experts models. While \cite{Liu_Stufken_Yang_2024} developed a surrogate matrix that could be targeted for the case of clusterwise linear regression, this is, at least from a theoretical perspective, more difficult for general mixture-of-experts models. Fourth, paraphrasing Box, all models are wrong, though some are useful. However, it may be uncertain which model is useful. Thus, expanding the methodology to selecting subdata that is efficient for a family of models will be very helpful. Fifth, the focus in this paper has been on $A$- and $D$-optimality. These are criteria that are typically used when there is an interest in parameter estimation. Such subdata may not be the best when the primary focus is on prediction. Finally, the algorithm for IBOSS OBD is fairly complex, and the authors are planning to develop corresponding software packages to 
simplify the implementation and broaden the usage of this algorithm.

\bibliography{bibliography.bib} 

\appendix

\section{Technical appendices and supplementary material}
\subsection{Proofs of Theorem~1, Theorem~2, and Lemma~1}
\begin{proof}[Proof of Theorem~1]
In the context of optimal design, Theorem~3 in \cite{Welch1982} shows that identifying a D-optimal exact design is NP-hard. Using the notation in Equation~1, identifying an optimal exact design requires identification of $\boldsymbol{\delta} = (\delta_1, \ldots, \delta_N)$, subject to $\sum_{i=1}^N \delta_i = n$, such that the determinant of the information matrix is maximized. However, while the $\delta_i$'s must be 0 or 1 in the subdata selection problem, in the optimal exact design problem they can be any non-negative integer, as long as they add to $n$. 
 
The result of Theorem~1 follows from a connection between the optimal exact design problem and the subdata selection problem. If we replicate each of $N$ points $n$ times, then selecting optimal subdata of size $n$ from the resulting $n \times N$ points under the $D$-optimality criterion is equivalent to selecting a $D$-optimal exact design for the original population of $N$ points. Thus, if there were a polynomial time algorithm for the subdata selection problem, there would be one for the optimal exact design problem, which would contradict the result in  \cite{Welch1982}.
\end{proof}

\begin{proof}[Proof of Theorem 2]
We first show that $(I) \Rightarrow (II)$. Starting from the $\Phi$-optimal design $\xi^*$ in (I), let $(\mathcal{X}_1,\mathcal{X}_2,\mathcal{X}_3)$ be the corresponding partition of the points $\bx_i$ with weight equal to $\nu_i$, between $\nu_i$ and $\mu_i$, and equal to $\mu_i$, respectively. We first consider the case that $\mathcal{X}_2 \neq \emptyset$. If there is more than one point in $\mathcal{X}_2$, let $\bx_{c1}, \bx_{c2} \in \mathcal{X}_2$ be two distinct points. For a small positive $\epsilon$, define the new design $\xi$ that has the same weights as $\xi^*$, except that it increases the weight for $\bx_{c1}$ by $\epsilon$ and decreases that for $\bx_{c2}$ by the same amount. Thus, writing $w_i(\xi)$ for the new weights, 
\begin{enumerate}
    \item$ w_{i}(\xi)=w_{i}$ for $x_i \in \mathcal{X} \setminus \lbrace \bx_{c1},\bx_{c2}\rbrace$; 
    \item $w_{{c1}}(\xi)=w_{c1}+\epsilon$ for $\bx_{c1}$; and \item $w_{{c2}}(\xi)=w_{c2}-\epsilon$ for $\bx_{c2}$.
\end{enumerate}
Here, $\epsilon$ is taken to be small enough so that  $w_{{c1}}(\xi)\leq\mu_{c1}$ and $w_{{c2}}(\xi)\geq\nu_{c2}$. 
Using that $\Phi$ is linearly differentiable, we have
\begin{eqnarray*}
    F_\Phi(\xi^*, \xi)=&\sum_{\bx_i\in\mathcal{X}_1}F_\Phi(\xi^*;\bx_i)\nu_i+\sum_{\bx_i\in\mathcal{X}_3}F_\Phi(\xi^*;\bx_i)\mu_i+\sum_{\bx_i\in\mathcal{X}_2\setminus\lbrace \bx_{c1},\bx_{c2}\rbrace}F_\Phi(\xi^*;\bx_i)w_i \nonumber\\
    &+F_\Phi(\xi^*;\bx_{c1})(w_{c1}+\epsilon)+F_\Phi(\xi^*;\bx_{c2})(w_{c2}-\epsilon)\\
=&\sum_{\bx_i\in\mathcal{X}_1}F_\Phi(\xi^*;\bx_i)\nu_i+\sum_{\bx_i\in\mathcal{X}_3}F_\Phi(\xi^*;\bx_i)\mu_i+\sum_{\bx_i\in\mathcal{X}_2}F_\Phi(\xi^*;\bx_i)w_i\nonumber\\
    &+F_\Phi(\xi^*;\bx_{c1})\epsilon-F_\Phi(\xi^*;\bx_{c2})\epsilon\\
    =&F_\Phi(\xi^*,\xi^*)+\epsilon(F_\Phi(\xi^*;\bx_{c1})-F_\Phi(\xi^*;\bx_{c2}))\\
    =&\epsilon(F_\Phi(\xi^*;\bx_{c1})-F_\Phi(\xi^*;\bx_{c2})) \geq 0,
\end{eqnarray*}
where we used that $F_\Phi(\xi^*,\xi^*)=0$ and $\xi^*$ is $\Phi$-optimal. Thus, $F_\Phi(\xi^*;\bx_{c1})\geq F_\Phi(\xi^*;\bx_{c2})$. Similarly, 
$F_\Phi(\xi^*;\bx_{c2})\geq F_\Phi(\xi^*;\bx_{c1})$. Therefore, 
\begin{eqnarray}
\label{f_same}
    F_\Phi(\xi^*;\bx_{c2})= F_\Phi(\xi^*;\bx_{c1}).
\end{eqnarray}
Since $\bx_{c1}, \bx_{c2} \in \mathcal{X}_2$ are arbitrary,  
$F_\Phi(\xi^*;\bx_{i})$ is the same for all $\bx_i \in \mathcal{X}_2$, say $s$. Also by
$ F_\Phi(\xi^*, \xi^*)=0$, 
it follows that
\begin{eqnarray}
\label{s2}
   F_\Phi(\xi^*;\bx_{i})=s=\frac{-\sum_{\bx_i\in\mathcal{X}_1}F_\Phi(\xi^*;\bx_i)\nu_i-\sum_{\bx_i\in\mathcal{X}_3}F_\Phi(\xi^*;\bx_i)\mu_i}{1-(\sum_{\bx_i\in\mathcal{X}_1}\nu_i+\sum_{\bx_i\in\mathcal{X}_3}\mu_i)}
\end{eqnarray} 
for $\bx_i \in \mathcal{X}_2$. Since \eqref{s2} also holds if $\mathcal{X}_2$ consists of a single point, (II)(c) follows. Next we will show that (II)(b) holds. Select points $\bx_c \in \mathcal{X}_2$ and $\bx_u \in \mathcal{X}_3$. For a small positive $\epsilon$, we define another new design $\hat{\xi}$ with weights equal to the weights of $\xi^*$, except that the weight at $\bx_c$ is $\epsilon$ more and the weight at $\bx_u$ is $\epsilon$ less that that of $\xi^*$. Thus,
\begin{enumerate}
    \item$ w_{i}(\hat{\xi})=w_{i}$ for $\bx_i \in \mathcal{X}\setminus\lbrace \bx_c,\bx_u\rbrace$;
    \item $w_{{c}}(\hat{\xi})=w_{c}+\epsilon$ for $\bx_c$; and 
    \item $w_{{u}}(\hat{\xi})=w_{u}-\epsilon$ for $\bx_u$.
\end{enumerate}
Here, $\epsilon$ is small enough so that $w_{{c}}(\hat{\xi})\leq\mu_{c}$ and $w_{{u}}(\hat{\xi})\geq\nu_{u}$. Then,
\begin{eqnarray*}
    F_\Phi(\xi^*, \hat{\xi})=&\sum_{\bx_i\in\mathcal{X}_1}F_\Phi(\xi^*;\bx_i)\nu_i+\sum_{\bx_i\in\mathcal{X}_3\setminus\lbrace \bx_u\rbrace}F_\Phi(\xi^*;\bx_i)\mu_i+\sum_{\bx_i\in\mathcal{X}_2\setminus\lbrace \bx_{c}\rbrace}F_\Phi(\xi^*;\bx_i)w_i\nonumber\\ &+F_\Phi(\xi^*;\bx_{c})(w_{c}+\epsilon)+F_\Phi(\xi^*;\bx_{u})(w_{u}-\epsilon)\\
=&\sum_{\bx_i\in\mathcal{X}_1}F_\Phi(\xi^*;\bx_i)\nu_i+\sum_{\bx_i\in\mathcal{X}_3}F_\Phi(\xi^*;\bx_i)\mu_i+\sum_{\bx_i\in\mathcal{X}_2}F_\Phi(\xi^*;\bx_i)w_i\nonumber\\
    &+F_\Phi(\xi^*;\bx_{c})\epsilon-F_\Phi(\xi^*i;\bx_{u})\epsilon\\
    =&\epsilon(F_\Phi(\xi^*;\bx_{c})-F_\Phi(\xi^*;\bx_{u})) \geq 0,
\end{eqnarray*}
where we used again that $F_\Phi(\xi^*,\xi^*)=0$ and $\xi^*$ is $\Phi$-optimal. Thus $ F_\Phi(\xi^*;\bx_{u}) \leq F_\Phi(\xi^*;\bx_{c}) =s$
for arbitrary $\bx_u \in \mathcal{X}_3$. Similarly, we can show that $F_\Phi(\xi^*;\bx_{\ell}) \geq s$
for arbitrary $\bx_\ell \in \mathcal{X}_1$. Consequently (II)(b) follows.\\
Next we consider the case $\mathcal{X}_2 = \emptyset$. Select arbitrary points $\bx_\ell \in \mathcal{X}_1$ and $\bx_u \in \mathcal{X}_3$. For a small positive $\epsilon$, define a new design $\bar{\xi}$ with the same weights as $\xi^*$, except that the weight for $\bx_\ell$ is increased by $\epsilon$ and that for $\bx_u$ is decreased by $\epsilon$. Thus,
\begin{enumerate}
    \item$ w_{i}(\bar{\xi})=w_{i}$ for $\bx_i \in \mathcal{X}\setminus\lbrace \bx_\ell,\bx_u\rbrace$;
    \item$ w_{\ell}(\bar{\xi})=w_{\ell}+\epsilon$; and
    \item $w_{u}(\bar{\xi})=w_{u}-\epsilon$.
\end{enumerate}
Here, $\epsilon$ is small enough so that $w_{\ell}(\bar{\xi})\leq\mu_\ell$ and $w_{u}(\bar{\xi})\geq\nu_u$. Then,
\begin{eqnarray*}
 F_\Phi(\xi^*, \bar{\xi})=&\sum_{\bx_i\in\mathcal{X}_1}F_\Phi(\xi^*;\bx_i)\nu_i+\sum_{\bx_i\in\mathcal{X}_3}F_\Phi(\xi^*;\bx_i)\mu_i+\epsilon F_\Phi(\xi^*;\bx_\ell)-\epsilon F_\Phi(\xi^*;\bx_u)\\
 =&\epsilon(F_\Phi(\xi^*;x_\ell)-F_\Phi(\xi^*;x_u)) \geq 0.
\end{eqnarray*}
Since $\bx_\ell \in \mathcal{X}_1$ and $\bx_u \in \mathcal{X}_3$ are arbitrary, it follows that $\max\limits_{\bx_i\in\mathcal{X}_3}F_\Phi(\xi^*;\bx_i)\leq s \leq \min\limits_{\bx_i\in\mathcal{X}_1}F_\Phi(\xi^*;\bx_i)$.\\

Now we show that $(II) \Rightarrow (I)$.\\
By the convexity of $\Phi$, it suffices to show $F_\Phi(\xi^*, \xi)\geq 0$ for every competing design $\xi$. For any given $\xi \in \Xi$, we have $w_{i}(\xi)\geq w_i$ on $\mathcal{X}_1$ and $w_{i}(\xi)\leq w_i$ on $\mathcal{X}_3$. Thus we have
\begin{eqnarray*}
F_\Phi(\xi^*,\xi)=&\sum_{x_i\in\mathcal{X}_1}F_\Phi(\xi^*;\bx_i)w_{i}(\xi)+\sum_{\bx_i\in\mathcal{X}_2}F_\Phi(\xi^*;\bx_i)w_{i}(\xi)+\sum_{x_i\in\mathcal{X}_3}F_\Phi(\xi^*;\bx_i)w_{i}(\xi)\\
     \geq& \sum_{\bx_i\in\mathcal{X}_1}F_\Phi(\xi^*;\bx_i)w_{i}+ \min_{\bx_i\in\mathcal{X}_1}F_\Phi(\xi^*;\bx_i)\sum_{\bx_i\in\mathcal{X}_1}(w_{i}(\xi)-w_{i})\nonumber\\
    &+\sum_{\bx_i\in\mathcal{X}_3}F_\Phi(\xi^*;\bx_i)w_{i}+ \max_{\bx_i\in\mathcal{X}_3}F_\Phi(\xi^*;\bx_i)\sum_{\bx_i\in\mathcal{X}_3}(w_{i}(\xi)-w_{i})\nonumber\\
    &+\sum_{\bx_i\in\mathcal{X}_2}F_\Phi(\xi^*;\bx_i)w_{i}+\sum_{\bx_i\in\mathcal{X}_2}F_\Phi(\xi^*;\bx_i)(w_{i}(\xi)-w_{i})\\
    \geq& s(\sum_{\bx_i\in\mathcal{X}_1}(w_{i}(\xi)-w_{i})+\sum_{\bx_i\in\mathcal{X}_3}(w_{i}(\xi) - w_i)+\sum_{\bx_i\in\mathcal{X}_2}(w_{i}(\xi)-w_{i}))=0.
\end{eqnarray*}
\end{proof}

\begin{proof}[Proof of Lemma 1]
Let $n_i$ be the number of points in $\mathcal{X}^i$, $i=1,2,3$. Let $n_2 \ge 2$, which is the only case for which a proof is needed. The information matrix of $\xi_{\mathcal{P}}$ can be viewed as a function of \text{$\pmb{w}^\prime$}=$(w_1,\ldots,w_{n_2-1})$, where $w_i$ denotes the weight for $\bx_i\in \mathcal{X}_2$, $i=1,\ldots, n_2-1$. The weight of $\bx_{n_2}\in \mathcal{X}_2$ is $1-n_3/n-\sum_{i=1}^{n_2-1}w_i$. By Theorem 4 of \cite{Yang_Biedermann_Tang_JASA}, the optimal weights are critical points, that is, solutions of $\frac{\partial \Phi(\xi_{\mathcal{P}})}{\partial w_i}=0$, or boundary points, that is, $w_i=0$ or $1/n$. But in the latter case, the procedure would have moved the points to $\mathcal{X}_1$ or $\mathcal{X}_3$. Hence,  
	\begin{equation}
	\frac{\partial \Phi(\xi_{\mathcal{P}})}{\partial w_i}=F_{\Phi}(\xi_{\mathcal{P}};\bx_i)-F_{\Phi}(\xi_{\mathcal{P}};\bx_{n_2})=0, \ i=1,\ldots,n_2-1.
	\end{equation}
\end{proof}

\subsection{Proof of Theorem~3}
Before proving Theorem~3, we need the following lemma, which uses the notation and assumptions as in Theorem~3 and Algorithm~1.
\begin{lemma}\label{lemma2}
	 For $\Phi \in \Lambda$, suppose that there are a bounded approximate design $\xi_0 \in \Xi$ and $\delta>0$, so that, for an integer $t > 0$,
	$\Phi(\xi_{\mathcal{P}^{(t)}})-\Phi(\xi_0)>\delta$. Then one of the following statements must hold.
	\begin{itemize}
		\item[(a)] If $\mathcal{X}^{(t)}_2=\emptyset$, then 
		\begin{equation}
		F_{\Phi}(\xi_{\mathcal{P}^{(t)}};\bx_{1*})-F_{\Phi}(\xi_{\mathcal{P}^{(t)}};\bx_{3*})\leq -\delta. \label{lemma2:01}
		\end{equation}
		\item[(b)] If $\mathcal{X}^{(t)}_2\ne \emptyset$ and $\mathcal{X}^{(t)}_1= \emptyset$, then
			\begin{equation}
 F_{\Phi}(\xi_{\mathcal{P}^{(t)}};\bx_{3*})-s\geq \delta. \label{lemma2:02}
		   \end{equation}		
		\item[(c)] If $\mathcal{X}^{(t)}_2\ne \emptyset$ and $\mathcal{X}^{(t)}_3= \emptyset$, then
\begin{equation}
 F_{\Phi}(\xi_{\mathcal{P}^{(t)}};\bx_{1*})-s\leq -\delta. \label{lemma2:03}
\end{equation}	
		\item[(d)] If $\mathcal{X}^{(t)}_2\ne \emptyset$, $\mathcal{X}^{(t)}_1\ne  \emptyset$ and $\mathcal{X}^{(t)}_3\ne  \emptyset$, then one of the two following inequalities must hold:
\begin{equation}
\begin{split}
(i) F_{\Phi}(\xi_{\mathcal{P}^{(t)}};\bx_{1*})-s&\leq -\frac{\delta}{2},\\
(ii) F_{\Phi}(\xi_{\mathcal{P}^{(t)}};\bx_{3*})-s&\geq \frac{\delta}{2}. \label{lemma2:04}
\end{split}
\end{equation}	
Here $s=F_{\Phi}(\xi_{\mathcal{P}^{(t)}};\bx_i)$, $x_i \in \mathcal{X}^{(t)}_2$.
\end{itemize}
\end{lemma}

\begin{proof}
	For $0 \le \alpha \le 1$, consider the design $(1-\alpha) \xi_{\mathcal{P}^{(t)}} +\alpha \xi_0$. Using the convexity of $\Phi$, we have that
	$\Phi((1-\alpha)\xi_{\mathcal{P}^{(t)}}+\alpha \xi_0)\leq (1-\alpha)\Phi(\xi_{\mathcal{P}^{(t)}})+\alpha  \Phi(\xi_0)$. Hence,
	\begin{equation}\begin{split}
	F_{\Phi}(\xi_{\mathcal{P}^{(t)}},\xi_0)&=\lim_{\alpha \downarrow 0}\frac{1}{\alpha }\left(\Phi((1-\alpha)\xi_{\mathcal{P}^{(t)}}+\alpha \xi_0)-\Phi(\xi_{\mathcal{P}^{(t)}})\right)\\
	&\leq  \Phi(\xi_0)-\Phi(\xi_{\mathcal{P}^{(t)}})\\
	&<-\delta. \label{lemma2:2}
	\end{split}\end{equation}
	For each $\bx_i$, let $w_i$ and $w_{i}(\xi_0)$ be the weights in $\xi_{\mathcal{P}^{(t)}}$ and $\xi_0$, respectively. 
	By the definitions of $\mathcal{X}^{(t)}_1$ and $\mathcal{X}^{(t)}_3$, since $\Phi$ is linearly differentiable, we have
	\begin{equation}\begin{split}
	F_{\Phi}(\xi_{\mathcal{P}^{(t)}},\xi_0)=&\sum_{x_i\in \mathcal{X}^{(t)}_1}F_{\Phi}(\xi_{\mathcal{P}^{(t)}};\bx_i)w_{i}(\xi_0)+\sum_{\bx_i\in \mathcal{X}^{(t)}_2}F_{\Phi}(\xi_{\mathcal{P}^{(t)}};\bx_i)w_{i}(\xi_0)+\sum_{\bx_i\in \mathcal{X}^{(t)}_3}F_{\Phi}(\xi_{\mathcal{P}^{(t)}};\bx_i)w_{i}(\xi_0)\\
	\geq& \sum_{\bx_i\in \mathcal{X}^{(t)}_1}F_{\Phi}(\xi_{\mathcal{P}^{(t)}};\bx_i)w_{i} + F_{\Phi}(\xi_{\mathcal{P}^{(t)}};\bx_{1*})\sum_{\bx_i\in \mathcal{X}^{(t)}_1}(w_{i}(\xi_0)-w_i)\\
	&+\sum_{\bx_i\in \mathcal{X}^{(t)}_3}F_{\Phi}(\xi_{\mathcal{P}^{(t)}};\bx_i)w_{i} + F_{\Phi}(\xi_{\mathcal{P}^{(t)}};\bx_{3*})\sum_{\bx_i\in \mathcal{X}^{(t)}_3}(w_{i}(\xi_0)-w_i)\\
	&+\sum_{\bx_i\in \mathcal{X}^{(t)}_2}F_{\Phi}(\xi_{\mathcal{P}^{(t)}};\bx_i)w_{i} +  \sum_{\bx_i\in \mathcal{X}^{(t)}_2}F_{\Phi}(\xi_{\mathcal{P}^{(t)}};\bx_i)(w_{i}(\xi_0)-w_i)\\
= &  F_{\Phi}(\xi_{\mathcal{P}^{(t)}};\bx_{1*})\sum_{x_i\in \mathcal{X}^{(t)}_1}(w_{i}(\xi_0)-w_i)+ F_{\Phi}(\xi_{\mathcal{P}^{(t)}};\bx_{3*})\sum_{\bx_i\in \mathcal{X}^{(t)}_3}(w_{i}(\xi_0)-w_i)+ \\
&+ \sum_{\bx_i\in \mathcal{X}^{(t)}_2}s(w_{i}(\xi_0)-w_i),	
		 \label{lemma2:3}
	\end{split}\end{equation}
where we used that $F_{\Phi}(\xi_{\mathcal{P}^{(t)}}, \xi_{\mathcal{P}^{(t)}})=0$.
Notice that when $\mathcal{X}^{(t)}_2\ne \emptyset$, by Lemma 1, $s=F_{\Phi}(\xi_{S^{(t)}};\bx_i)$, $\bx_i \in \mathcal{X}^{(t)}_2$, is well-defined. 
By Equations~(\ref{lemma2:2}) and (\ref{lemma2:3}), for the right hand side (RHS) of (\ref{lemma2:3}) we have
\begin{equation}
\text{RHS of (\ref{lemma2:3})}< -\delta. \label{lemma2:4}
\end{equation}
Next, we consider each case in the statement of the lemma separately, and will arrive at a contradiction with Equation~(\ref{lemma2:4}) if a case does not hold. 

For Case (a), suppose that Equation~(\ref{lemma2:01}) does not hold, so that $F_{\Phi}(\xi_{\mathcal{P}^{(t)}};\bx_{1*})-F_{\Phi}(\xi_{\mathcal{P}^{(t)}},\bx_{3*})> -\delta$. 
This implies that
\begin{equation*}
\begin{split}
\text{RHS of (\ref{lemma2:3})}&\geq F_{\Phi}(\xi_{\mathcal{P}^{(t)}};\bx_{3*})\sum_{\bx_i\in \mathcal{X}}(w_{i}(\xi_0)-w_i)-\delta \sum_{\bx_i\in \mathcal{X}^{(t)}_1}(w_{i}(\xi_0)-w_i)\\
 &\geq -\delta. 
 \end{split}
\end{equation*}
This contradicts Equation~(\ref{lemma2:4}), so that Case (a) must hold.

For Case (b), suppose that Equation~(\ref{lemma2:02}) does not hold, so that $F_{\Phi}(\xi_{\mathcal{P}^{(t)}};\bx_{3*})-s< \delta$. 
This implies that 
\begin{equation*}
\begin{split}
\text{RHS of (\ref{lemma2:3})}&=s\sum_{\bx_i\in \mathcal{X}}(w_{i}(\xi_0)-w_i) +\left(F_{\Phi}(\xi_{\mathcal{P}^{(t)}};\bx_{3*})-s\right)\sum_{\bx_i\in \mathcal{X}^{(t)}_3}(w_{i}(\xi_0)-w_i)\\
&\geq \delta \sum_{\bx_i\in \mathcal{X}^{(t)}_3}(w_{i}(\xi_0)-w_i)\geq -\delta. 
\end{split}
\end{equation*}
This again contradicts Equation~(\ref{lemma2:4}), so that Case (b) must hold.

Next, for Case (c), suppose that Equation~(\ref{lemma2:03}) does not hold, so that $F_{\Phi}(\xi_{\mathcal{P}^{(t)}};\bx_{1*})-s> -\delta$. 
This implies that
\begin{equation*}
\begin{split}
\text{RHS of (\ref{lemma2:3})}&=s\sum_{\bx_i\in \mathcal{X}}(w_{i}(\xi_0)-w_i) +\left(F_{\Phi}(\xi_{\mathcal{P}^{(t)}};\bx_{1*})-s\right)\sum_{\bx_i\in \mathcal{X}^{(t)}_1}(w_{i}(\xi_0)-w_i)\\
&\geq -\delta \sum_{\bx_i\in \mathcal{X}^{(t)}_1}(w_{i}(\xi_0)-w_i)\geq -\delta. 
\end{split}
\end{equation*}
This is once more a contradiction to Equation~(\ref{lemma2:4}), this time implying that Case (c) must hold.
	
Finally, for Case (d), suppose that neither of the two inequalities in (\ref{lemma2:04}) holds, so that 
\begin{equation}
\begin{split}
(i) F_{\Phi}(\xi_{\mathcal{P}^{(t)}};\bx_{1*})-s&> -\frac{\delta}{2}, \text{~and}\\
(ii) F_{\Phi}(\xi_{\mathcal{P}^{(t)}};\bx_{3*})-s&< \frac{\delta}{2}. \label{lemma2:5}
\end{split}
\end{equation}	
This implies that
\begin{equation}
\begin{split}
\text{RHS of (\ref{lemma2:3})}&=s\sum_{\bx_i\in \mathcal{X}}(w_{i}(\xi_0)-w_i) +
\left(F_{\Phi}(\xi_{\mathcal{P}^{(t)}};\bx_{1*})-s\right)\sum_{\bx_i\in \mathcal{X}^{(t)}_1}(w_{i}(\xi_0)-w_i)\\ &+\left(F_{\Phi}(\xi_{\mathcal{P}^{(t)}};\bx_{3*})-s\right)\sum_{\bx_i\in \mathcal{X}^{(t)}_3}(w_{i}(\xi_0)-w_i)\\
\ge & -\frac{\delta}{2} \sum_{\bx_i\in \mathcal{X}^{(t)}_1}(w_{i}(\xi_0)-w_i)
+\frac{\delta}{2} \sum_{\bx_i\in \mathcal{X}^{(t)}_2}(w_{i}(\xi_0)-w_i)\\
\ge & -\delta.
\label{lemma2:6}
\end{split}
\end{equation}
This contradiction to Equation~(\ref{lemma2:4}) implies that Case (d) must hold.
\end{proof}

\begin{lemma}\label{lemma3}
	Under $\Phi$ optimality, i.e., 
\begin{equation}
    \Phi(\xi)=\begin{cases}\log \text{det}(I_{\boldsymbol{\theta}}^{-1}(\xi)), & D-\text{optimality}\\
  Tr(I_{\boldsymbol{\theta}}^{-p}(\xi)), & A-\text{optimality}      
    \end{cases},
\end{equation}
for designs $\xi_1$, $\xi_2$ and $\xi_3$ and  positive value $\gamma$ for which $\xi = \xi_1 + \gamma(\xi_2 - \xi_3)$ is a design and $\Phi(\xi)$ is bounded, then $\Phi(\xi)$  has a bounded second derivative respective to $\gamma$.
\end{lemma}

\begin{proof}
We first consider the case $p=0$; the case $p>0$ is analogous.  
When $p=0$,
\[
\frac{\partial^2 \Phi(\xi)}{\partial \gamma^2}
= \operatorname{Tr}\!\left(
I_{\boldsymbol{\theta}}^{-1}(\xi) 
\Delta
I_{\boldsymbol{\theta}}^{-1}(\xi) 
\Delta
\right),
\quad
\Delta := I_{\boldsymbol{\theta}}(\xi_2) - I_{\boldsymbol{\theta}}(\xi_3).
\]
Since $\Phi(\xi)$ is bounded, the eigenvalues (hence all entries) of $I_{\boldsymbol{\theta}}^{-1}(\xi)$ are bounded.  
Each $I_{\boldsymbol{\theta}}(\xi_j)$, $j=2,3$, is a convex combination of finitely many well-defined information matrices, so their entries are bounded; therefore $\Delta$ has bounded entries.  
This establishes the claim.

\end{proof}

\begin{proof}[Proof of Theorem 3]
Let $\xi^*=argmin_{\xi\in \Xi}\Phi\left({\xi}\right) $
and consider $\Phi(\xi[\bx_+, \bx_-,\gamma])$, where $\Xi$ and $\xi[\bx_+, \bx_-,\gamma]$ defined in Section 3.2.
Then $\Phi(\xi[\bx_+, \bx_-,\gamma])$ is infinitely differentiable with respect to $\gamma$. In the iteration stage, the $\Phi$-values of designs decrease. Thus, they are uniformly bounded by the $\Phi$-value of the design at the beginning of the iteration stage, say $\Phi(\xi_0)$. By Assumption 1, we have that 
$\Phi(\xi[\bx_+, \bx_-,\gamma])<C\Phi(\xi)$ for some constant $C>0$. Consequently, $\Phi(\xi[\bx_+, \bx_-,\gamma])<C\Phi(\xi_0)$. 
By Lemma \ref{lemma3}, there exists $K<\infty$, such that
	\begin{equation}\begin{split}
	\frac{\partial^2\Phi(\xi[\bx_+, \bx_-,\gamma])}{\partial \gamma^2}=K. \label{convg:1}
	\end{split}\end{equation}
Notice that $\Phi(\xi_{\mathcal{P}^{(t)}})$ is a decreasing nonnegative function of $t$ and bounded by $\Phi({\xi^*})$. Thus $\lim_{t\rightarrow \infty}\Phi(\xi_{\mathcal{P}^{(t)}})$  exists. We shall show that
	\begin{equation}\begin{split}
	\lim_{t\rightarrow \infty}\Phi(\xi_{\mathcal{P}^{(t)}})=\Phi({\xi^*}). \label{convg:2}
	\end{split}\end{equation}
If Equation~(\ref{convg:2}) does not hold, then, by the monotonicity of $\Phi(\xi_{\mathcal{P}^{(t)}})$, there exists a $\delta>0$ such that
$\Phi(\xi_{\mathcal{P}^{(t)}})-\Phi({\xi^*})>\delta $ for all $t$. We will show that this leads to a contradiction by considering the following four exhaustive scenarios for the sequence of partitions $\mathcal{P}^{(t)}$:
\begin{itemize}
	\item[(i)] $\mathcal{X}^{(t)}_2=\emptyset$ for infinitely many values of $t$;
	\item[(ii)] $\mathcal{X}^{(t)}_2\ne \emptyset$ and $\mathcal{X}^{(t)}_1= \emptyset$ for infinitely many values of $t$;
	\item[(iii)] $\mathcal{X}^{(t)}_2\ne \emptyset$ and $\mathcal{X}^{(t)}_3= \emptyset$ for infinitely many values of $t$;
	\item[(iv)] $\mathcal{X}^{(t)}_2\ne \emptyset$, $\mathcal{X}^{(t)}_1\ne  \emptyset$, and $\mathcal{X}^{(t)}_3\ne \emptyset$ for infinitely many values of $t$.	 
\end{itemize}	

If Scenario (i) holds, then there must exist an infinite subsequence of $\mathcal{P}^{(t)}$ such that $\mathcal{X}^{(t)}_2=\emptyset$. For simplicity of notation, we index the subsequence again by $t$, and have a sequence $\mathcal{P}^{(t)}$ so that $\mathcal{X}^{(t)}_2=\emptyset$ for all $t$. Define the new design $\xi_{\mathcal{P}^{(t)}}[\bx_{1*}, \bx_{3*},\gamma]$, which belongs to $\Xi$ for $\gamma < 1/n$. Note that $\bx_{1*}$ and $\bx_{3*}$ depend on $t$, but we surpress this in the notation for simplicity. By the definition of $\mathcal{P}^{(t+1)}$, we have that
\begin{equation}\begin{split}
\Phi(\xi_{\mathcal{P}^{(t+1)}})\leq \Phi(\xi_{\mathcal{P}^{(t)}}[\bx_{1*}, \bx_{3*},\gamma]). \label{convg:3}
\end{split}\end{equation}
Using a Taylor series expansion near $\gamma = 0$, and applying Equations~(\ref{lemma2:01}) and (\ref{convg:1}), it can be shown that
\begin{equation}\begin{split}
\Phi(\xi_{\mathcal{P}^{(t)}}[x_{1*}, x_{3*},\gamma])=&\Phi(\xi_{\mathcal{P}^{(t)}})+\left(F_{\Phi}(\xi_{\mathcal{P}^{(t)}},\bx_{1*})- F_{\Phi}(\xi_{\mathcal{P}^{(t)}},\bx_{3*})\right)\gamma\\
&+ \frac{1}{2}\gamma^2 \left.\frac{\partial^2\Phi(\xi[x_{1*}, x_{3*},\gamma])}{\partial \gamma^2}\right|_{\gamma=\gamma'}\\
\leq& \Phi(\xi_{\mathcal{P}^{(t)}})-\delta\gamma+ \frac{1}{2}K\gamma^2, \label{convg:4}
\end{split}\end{equation}
where $\gamma'\in[0,\gamma]$. If $K>n\delta$, let $\gamma=\frac{\delta}{K}$. By Equation~(\ref{convg:4}),  we have that
\begin{equation}\begin{split}
\Phi(\xi_{\mathcal{P}^{(t)}}[\bx_{1*}, \bx_{3*},\frac{\delta}{K}])-\Phi(\xi_{\mathcal{P}^{(t)}})\leq -\frac{\delta^2}{2K}. \label{convg:5}
\end{split}\end{equation}
By (\ref{convg:3}) and (\ref{convg:5}), we have, for all $t\geq 0$,
\begin{equation}\begin{split}
\Phi(\xi_{\mathcal{P}^{(t+1)}})-\Phi(\xi_{\mathcal{P}^{(t)}})\leq -\frac{\delta^2}{2K}. \label{convg:6}
\end{split}\end{equation}
Since Equation~(\ref{convg:6}) holds for all $t$, it implies that $\lim_{t\rightarrow \infty}\Phi(\xi_{\mathcal{P}^{(t)}})=-\infty$, which contradicts the fact
that $\Phi(\xi_{\mathcal{P}^{(t)}})$ is bounded by $\Phi({\xi^*})$. Similar arguments can be applied to the case when $K\leq n\delta$, in which case we let $\gamma=1/n$.

If Scenario (ii) holds, then, again by focusing on a subsequence of $\mathcal{P}^{(t)}$, we can assume that $\mathcal{X}^{(t)}_2\ne \emptyset$ and $\mathcal{X}^{(t)}_1= \emptyset$ for all $t$. Define $\bx_{2*}$ as a point in $\mathcal{X}_2^{(t)}$ with the smallest weight, say $w_{min}^{(t)}$. Note that $\bx_{2*}$ depends on $t$, which we again suppress in the notation for simplicity.
Note also that $0<w_{min}^{(t)}<1/n$, so that ${\liminf}_{t\rightarrow \infty}w_{min}^{(t)}$ exists. Furthermore, ${\liminf}_{t\rightarrow \infty}w_{min}^{(t)} < 1/n$, because otherwise, for any given $\epsilon>0$, for large enough $t$ we would have that $w_i>1/n-\epsilon$ for all $\bx_i\in \mathcal{X}^{(t)}_2$. Since points that are not in $\mathcal{X}^{(t)}_2$ must be in $\mathcal{X}^{(t)}_3$ for this scenario, it would contradict $\sum_{i=1}^N w_i=1$. Define 
$\gamma_0={\liminf}_{t\rightarrow \infty}w_{min}^{(t)}$. 
Consider the design 
$\xi_{\mathcal{P}^{(t)}}[\bx_{2*}, \bx_{3*},\gamma]$, where $\gamma\in[0,(1/n-\gamma_0)/2]$. From a Taylor series expansion as in Equation~(\ref{convg:4}), and using Equation~(\ref{lemma2:02}), the remainder of the argument is similar to that for design $\xi_{\mathcal{P}^{(t)}}[\bx_{1*}, \bx_{3*},\gamma)$ in Scenario (i).
 
We skip the proof for Scenario (iii) since this case is similar to Scenario (ii).

If Scenario (iv) holds, then (i) or (ii) of Equation~(\ref{lemma2:04}) must hold for infinitely many values of $t$. Without loss of generality we assume that (i) since the proof for (ii) is similar. By taking a subsequence, we can again assume that it holds for all $t$. 
For this scenario, we define $\bx_{2*}$ (with its dependence on $t$ again suppressed in the notation) as a point in $\mathcal{X}_2^{(t)}$ with the largest weight, say $w_{max}^{(t)}$.
Then, $0<w_{max}^{(t)}<1/n$ and ${\limsup}_{t\rightarrow \infty}w_{max}^{(t)}$ exists.

Define $\gamma_1={\limsup}_{t\rightarrow \infty}w_{max}^{(t)}$, and assume that it is positive. Consider the design $\xi_{\mathcal{P}^{(t)}}[\bx_{1*}, \bx_{2*},\gamma]$ for $\gamma\in[0,\gamma_1/2]$. 
From a Taylor series expansion as in Equation~(\ref{convg:4}), and using (i) of Equation~(\ref{lemma2:04}), the remainder of the argument is similar to that for design $\xi_{\mathcal{P}^{(t)}}(x_{1*}, x_{3*},\gamma)$ in the proof for Scenario (i).

It remains to consider the case $\gamma_1=0$. Then, for any $\epsilon>0$, for large enough $t$, $w_i<\epsilon$ for all $\bx_i\in \mathcal{X}^{(t)}_2$. With $n_i$ denoting the number of points in $\mathcal{X}^{(t)}_i$, $i=1,2,3$, we have, for large enough $t$, that $n_3/n\leq 1<n_3/n+n_2\epsilon$. Since this holds for any $\epsilon>0$, it implies that $n_3/n= 1$ for large enough $t$. But that implies that $n_2=0$, which contradicts that $\mathcal{X}^{(t)}_2\ne \emptyset$ for all $t$. 
\end{proof}

\subsection{Proofs of Theorem~4 and Theorem~5}
\begin{proof}[Proof of Theorem 4]
From the convexity of $\Phi$, for any design $\xi=\{(\bx_i,w_{i}(\xi)), i=1,\ldots,N\}$, we have that 
\begin{eqnarray*}\begin{split}
\Phi\left((1-\alpha)\xi_{\epsilon}+\alpha \xi\right)-\Phi (\xi_{\epsilon})&\leq (1-\alpha)\Phi(\xi_{\epsilon})+\alpha \Phi (\xi)-\Phi (\xi_{\epsilon})\\
&=\alpha \left(\Phi (\xi)-\Phi (\xi_{\epsilon})\right).
\end{split}
\end{eqnarray*}
This implies that 
\begin{eqnarray}\label{speed2}
F_\Phi(\xi_{\epsilon},\xi) \leq \Phi (\xi)-\Phi (\xi_{\epsilon}).
\end{eqnarray}
Next we shall show that  $F_\Phi(\xi_{\epsilon},\xi)>-\epsilon$. 
Using $w_i$ to denote the weight of $\bx_i$ in $\xi_{\epsilon}$, note that $w_{i}(\xi)\geq w_{i}$ on $\mathcal{X}_1$ and $w_{i}(\xi)\leq w_{i}$ on $\mathcal{X}_3$. There are two cases we need to consider (i)  $\mathcal{X}_2\ne \emptyset$ and (ii) $\mathcal{X}_2= \emptyset$. 

Case (i) $\mathcal{X}_2\ne \emptyset$, then we have
\begin{eqnarray*}
\begin{split}
F_\Phi(\xi_{\epsilon},\xi) =&\sum_{x_i\in\mathcal{X}_1}F_\Phi(\xi_{\epsilon};\bx_i)w_{i}(\xi)+\sum_{\bx_i\in\mathcal{X}_2}F_\Phi(\xi_{\epsilon};\bx_i) w_{i}(\xi)+\sum_{x_i\in\mathcal{X}_3}F_\Phi(\xi_{\epsilon};\bx_i)w_{i}(\xi)\\
     \geq& 
\sum_{x_i\in\mathcal{X}_1}F_\Phi(\xi_{\epsilon};\bx_i)w_{i}+ \min_{\bx_i\in\mathcal{X}_1}F_\Phi(\xi_{\epsilon};\bx_i)\sum_{\bx_i\in\mathcal{X}_1}(w_{i}(\xi)-w_{i})\\
&+\sum_{\bx_i\in\mathcal{X}_3}F_\Phi(\xi_{\epsilon};\bx_i)w_{i}- \max_{\bx_i\in\mathcal{X}_3}F_\Phi(\xi_{\epsilon};\bx_i)\sum_{\bx_i\in\mathcal{X}_3}(w_{i}-w_{i}(\xi))\\
&+\sum_{\bx_i\in\mathcal{X}_2}F_\Phi(\xi_{\epsilon};\bx_i)w_{i}+(s-\frac{\epsilon}{2})\sum_{\bx_i\in\mathcal{X}_2 \text{ and } w_{i}(\xi)>w_{i}} (w_{i}(\xi)-w_{i})\\
&+(s+\frac{\epsilon}{2})\sum_{\bx_i\in\mathcal{X}_2 \text{ and } w_{i}(\xi)\leq w_{i} }(w_{i}(\xi)-w_{i})\\
    \geq& s\left(\sum_{\bx_i\in\mathcal{X}_1}(w_{i}(\xi)-w_{i})+\sum_{\bx_i\in\mathcal{X}_3}(w_{i}(\xi)-w_{i})+\sum_{\bx_i\in\mathcal{X}_2}(w_{i}(\xi)-w_{i})\right)\\
    &-\frac{\epsilon}{2}\left(\sum_{\bx_i\in\mathcal{X}_1}(w_{i}(\xi)-w_{i})+ \sum_{\bx_i\in\mathcal{X}_3}(w_{i}-w_{i}(\xi)) \right)\\
    &-\frac{\epsilon}{2}\left(\sum_{\bx_i\in\mathcal{X}_2 \text{ and } w_{i}(\xi)>w_{i} }(w_{i}(\xi)-w_{i}) +
    \sum_{\bx_i\in\mathcal{X}_2 \text{ and } w_{i}(\xi)>w_{i} }(w_{i}-w_{i}(\xi))\right)\\
    &\geq -\epsilon.
\end{split}
\end{eqnarray*}

Case (ii) $\mathcal{X}_2= \emptyset$,  then we have
\begin{eqnarray*}
\begin{split}
F_\Phi(\xi_{\epsilon},\xi) =&\sum_{x_i\in\mathcal{X}_1}F_\Phi(\xi_{\epsilon};\bx_i)w_{i}(\xi)+\sum_{x_i\in\mathcal{X}_3}F_\Phi(\xi_{\epsilon};\bx_i)w_{i}(\xi)\\
     \geq& 
\sum_{x_i\in\mathcal{X}_1}F_\Phi(\xi_{\epsilon};\bx_i)w_{i}+ \min_{\bx_i\in\mathcal{X}_1}F_\Phi(\xi_{\epsilon};\bx_i)\sum_{\bx_i\in\mathcal{X}_1}(w_{i}(\xi)-w_{i})\\
&+\sum_{\bx_i\in\mathcal{X}_3}F_\Phi(\xi_{\epsilon};\bx_i)w_{i}- \max_{\bx_i\in\mathcal{X}_3}F_\Phi(\xi_{\epsilon};\bx_i)\sum_{\bx_i\in\mathcal{X}_3}(w_{i}-w_{i}(\xi))\\
    \geq& \max_{\bx_i\in\mathcal{X}_3}
  \left(\sum_{\bx_i\in\mathcal{X}_1}(w_{i}(\xi)-w_{i})+\sum_{\bx_i\in\mathcal{X}_3}(w_{i}(\xi)-w_{i})\right)
    -\epsilon \sum_{\bx_i\in\mathcal{X}_1}(w_{i}(\xi)-w_{i})\\
    \geq & -\epsilon.
\end{split}
\end{eqnarray*}
For both cases, we have $F_\Phi(\xi_{\epsilon},\xi)>-\epsilon$. By (\ref{speed2}),  $\xi_{\epsilon}$ is an $\epsilon$-approximation of the $\Phi$-optimal design $\xi^*$. 
\end{proof}

\begin{proof}[Proof of Theorem 5]
We shall study $\Phi(\xi_{\mathcal{P}^{(t)}})-\Phi(\xi_{\mathcal{P}^{(t+1)}})$ when, for some $t$, Condition $(b)$ or Condition $(c)$ in Theorem~4 is not satisfied. We use notation defined in the proof of Theorem~3, such as $\xi^*$, $\xi[\boldsymbol{x}_{+},\boldsymbol{x}_{-},\gamma]$, and $K$. First consider the case $\mathcal{X}_2^{(t)} = \emptyset$. Then 	
\begin{equation}
		F_{\Phi}(\xi_{\mathcal{P}^{(t)}};\bx_{1*})-F_{\Phi}(\xi_{\mathcal{P}^{(t)}};\bx_{3*})\leq -\epsilon. \label{speed3}
\end{equation}
For design $\xi_{\mathcal{P}^{(t)}}[\bx_{1*}, \bx_{3*},\gamma]$, from the definition of $\mathcal{P}^{(t+1)}$, it follows that 
\begin{equation}\begin{split}
\Phi(\xi_{\mathcal{P}^{(t+1)}})\leq \Phi(\xi_{\mathcal{P}^{(t)}}[\bx_{1*}, \bx_{3*},\gamma]). \label{speed4}
\end{split}\end{equation}
Using again a Taylor series expansion near $\gamma = 0$, and applying Equations~(\ref{speed3}) and (\ref{convg:1}), we can show that
\begin{equation}\begin{split}
\Phi(\xi_{\mathcal{P}^{(t)}}[x_{1*}, x_{3*},\gamma])=&\Phi(\xi_{\mathcal{P}^{(t)}})+\left(F_{\Phi}(\xi_{\mathcal{P}^{(t)}};\bx_{1*})- F_{\Phi}(\xi_{\mathcal{P}^{(t)}};\bx_{3*})\right)\gamma\\
&+ \frac{1}{2} \left. \gamma^2 \frac{\partial^2\Phi(\xi[x_{1*}, x_{3*},\gamma])}{\partial \gamma^2}\right|_{\gamma=\gamma'}\\
\leq& \Phi_{p}(\xi_{\mathcal{P}^{(t)}})-\epsilon\gamma+ \frac{1}{2}K\gamma^2, \label{speed5}
\end{split}\end{equation}
where $\gamma'\in[0,\gamma]$. If $K>n\epsilon$, let $\gamma=\frac{\epsilon}{K}$. By Equation~(\ref{speed5}), we have that
\begin{equation}\begin{split}
\Phi(\xi_{\mathcal{P}^{(t)}}[x_{1*}, x_{3*},\gamma])-\Phi(\xi_{\mathcal{P}^{(t)}})\leq -\frac{\epsilon^2}{2K}. \label{speed6}
\end{split}\end{equation}
If $K\leq n\epsilon$, let $\gamma=1/n$. Then we have 
\begin{equation}\begin{split}
\Phi(\xi_{\mathcal{P}^{(t)}}[x_{1*}, x_{3*},\gamma])-\Phi(\xi_{\mathcal{P}^{(t)}})\leq -\frac{\epsilon}{2n}. \label{speed7}
\end{split}\end{equation}
By Equations~(\ref{speed4}), (\ref{speed6}), and (\ref{speed7}), we see that 
\begin{equation}\label{speed8}
\Phi(\xi_{\mathcal{P}^{(t+1)}})-\Phi(\xi_{\mathcal{P}^{(t)}})<\begin{cases}
-\frac{\epsilon^2}{2K} & \text{ when }  K>n\epsilon \\
-\frac{\epsilon}{2n} & \text{ when } K\leq n\epsilon.
\end{cases}
\end{equation}
Next, consider the case $\mathcal{X}_2^{(t)} \ne \emptyset$. Notice that Algorithm~2 guarantees that the weights of points $\bx_i\in \mathcal{X}_2^{(t)}$ are optimized strictly between 0 and $1/n$. By Lemma~1 and the setting of Algorithm~2, $(i)$ and $(ii)$ of Condition $(c)$ in Theorem 4 are automatically satisfied. There are two possibilities when $(iii)$ of Condition $(c)$ in Theorem 4 is not satisfied:
\begin{itemize}
    \item[(i)] $F_{\Phi}(\xi_{\mathcal{P}^{(t)}};\bx_{1*})-s< -\frac{\epsilon}{2}  \text{ when } \mathcal{X}_1^{(t)} \ne \emptyset$
    \item[(ii)] $F_{\Phi}(\xi_{\mathcal{P}^{(t)}};\bx_{3*})-s> \frac{\epsilon}{2} \text{ when } \mathcal{X}_3^{(t)} \ne \emptyset.$
\end{itemize}
If (i) holds, let $\bx_{2*} \in \mathcal{X}_2^{(t)}$ be a point with the largest weight of the points in $\mathcal{X}_2^{(t)}$.
Since the weights must sum to 1, the weight of $\bx_{2*}$ must be at least $\frac{1}{nN}$. Consider the design $\xi_{\mathcal{P}^{(t)}}[\bx_{1*}, \bx_{2*},\gamma]$ for a small enough $\gamma$. Similarly as for $\xi_{\mathcal{P}^{(t)}}[\bx_{1*}, \bx_{3*},\gamma]$, we have that
\begin{equation}\label{speed9}
\Phi(\xi_{\mathcal{P}^{(t+1)}})-\Phi(\xi_{\mathcal{P}^{(t)}})<\begin{cases}
-\frac{\epsilon^2}{8K} & \text{ when }  K>nN\epsilon/2 \\
-\frac{\epsilon}{4nN} & \text{ when } K\leq nN\epsilon/2.
\end{cases}
\end{equation}
If (ii) holds, let $\bx_{2*} \in \mathcal{X}_2^{(t)}$ be a point with the smallest weight of the points in $\mathcal{X}_2^{(t)}$. The weight of $\bx_{2*}$ is at most $1/(2n)$.
For small enough $\gamma$, consider the new design $\xi_{\mathcal{P}^{(t)}}[\bx_{2*}, \bx_{3*},\gamma]$. Similarly as for $\xi_{\mathcal{P}^{(t)}}[\bx_{1*}, \bx_{3*},\gamma]$, we have
\begin{equation}\label{speed10}
\Phi(\xi_{\mathcal{P}^{(t+1)}})-\Phi(\xi_{\mathcal{P}^{(t)}})<\begin{cases}
-\frac{\epsilon^2}{8K} & \text{ when }  K> n\epsilon \\
-\frac{\epsilon}{8n} & \text{ when } K\leq n\epsilon.
\end{cases}
\end{equation}
Thus, when the conditions
in Theorem~5 are not satisfied for some $t$, then we have by Equations~(\ref{speed8}), (\ref{speed9}), and (\ref{speed10})
\begin{equation}\label{speed11}
\Phi(\xi_{\mathcal{P}^{(t+1)}})-\Phi(\xi_{\mathcal{P}^{(t)}})<max(-\frac{\epsilon^2}{8K}, -\frac{\epsilon}{4nN}).
\end{equation}
With
\begin{equation*}
\xi^*=argmin_{\xi\in \Xi}\Phi\left({\xi}\right),
\end{equation*} 
for a given $\epsilon>0$, after at most 
\begin{equation}\label{speed12}
(\Phi(\xi_{\mathcal{P}^{(0)}})-\Phi(\xi^*)) \times \max(\frac{8K}{\epsilon^2}, \frac{4nN}{\epsilon})
\end{equation}
iterations,  $\mathcal{P}^{(t)}$ converges to a $\epsilon$-approximation of $\xi^*$. Notice that $\Phi(\xi_{\mathcal{P}^{(0)}})-\Phi(\xi^*)$ is finite and $K$ is constant; this implies that the algorithm's run time is bounded by a polynomial function in $n$ and $N$. Consequently, the desired conclusion follows. 
\end{proof}

\section{Further considerations for the main algorithm} \label{implementation} 
Theorems~\ref{convg} and \ref{speed} provide assurances for the performance of the main algorithm. But choices that can impact speed still need to be made. This includes selecting initial subdata $S$ of size $n$ in line~\ref{alg:startS} of Algorithm~\ref{alg:main} and an efficient algorithm to derive optimal weights in line~\ref{alg:optw} of Algorithm~\ref{alg:main}. 

\subsection{Selecting initial subdata} Starting Algorithm~\ref{alg:main} in line~\ref{alg:startS} with efficient subdata will generally result in a shorter running time to reach convergence. We recommend that an adaptation of the subdata selection method in \cite{Wang2019Information-BasedRegression} and \cite{Cheng2020Information-basedRegression} be used to select the initial subdata. The models considered by these authors use a predictor of the form $\beta_0+\boldsymbol{\beta}^T \boldsymbol{x}$, and their recommendation is roughly to select for each $j$, $1\leq j\leq p$, the $n/(2p)$ points $\bx_i$ with the largest values for the $j$th feature and the $n/(2p)$ points $\bx_i$ with the smallest values for the $j$th feature. They demonstrate desirable theoretical results for this strategy, called the IBOSS strategy, and we will adapt it here for more general statistical models. 

More generally, the information at a point $\bx$ will be proportional to $I_{\boldsymbol{\theta}}(\bx)=f(\bx)\times f^T(\bx)$. Here, $f(\bx)$ is a vector for a univariate response variable. If the response variable is multivariate, $f(\bx)$ is a matrix. In that case, we vectorize this matrix, and call it again $f(\bx)$. 

Form the matrix $(f(\bx_1),\ldots,f(\bx_N))$, remove its constant rows (if any), and let $F$ be the resulting matrix with $p_2$ rows. We propose the following algorithm to generate initial subdata $S$. 
\begin{itemize}
    \item[(1)] For the first row of $F$, select $\lfloor n/(2p_2) \rfloor$ points with the largest values and $\lfloor n/(2p_2) \rfloor$ points with the smallest values in this row.
    \item[(2)] For the next row, exclude points that have already been selected, and select $\lfloor n/(2p_2)\rfloor$ points with the largest values and $\lfloor n/(2p_2) \rfloor$ points with the smallest values in this row.
    \item[(3)] Repeat step (2) until all $p_2$ rows have been used.
    \item[(4)] If the total number of points selected is less than $n$, randomly select $n-2p_2\lfloor n/(2p_2) \rfloor$ points from the remaining points.
    \item[(5)] Use the selected $n$ points as the initial subdata for line~\ref{alg:startS} in Algorithm~\ref{alg:main}.
\end{itemize}
The computational complexity of this algorithm is $O(Np_2)$. Although the subdata generated by this fast algorithm may not be highly efficient for all models, it provides a starting point for Algorithm~\ref{alg:main} that tends to be far better than using randomly selected subdata. 

\subsection{Deriving optimal weights} Common algorithms for deriving optimal weights, such as the multiplicative algorithm~\citep{Yu_StatComp}, first-order algorithm~\citep{Ahipasaoglu_StatComp}, and lift-one algorithm~\citep{Huang_Tong_Yang_Sinica}, create a sequence of weights that converges to the optimal weights. These algorithms improve at every iteration and have a closed-form iterative formula for updating the weights. A disadvantage is that they tend to be rather slow, especially for large $N$.  

A different idea for finding optimal weights was proposed in \cite{Yang_Biedermann_Tang_JASA}. With constraints $w_i\geq 0$ and $\sum_{i=1}^Nw_i=1$, the authors proved that optimal weights correspond to a boundary point, i.e., $w_i=0$ for at least one $i$, or to a critical point for the objective function, i.e.,  $\frac{\partial \Phi(\xi_{\mathcal{P}^{(t)}})}{\partial w_i}=0$ for all $i$  (see Theorem 4 in \cite{Yang_Biedermann_Tang_JASA}), where $\mathcal{P}^{(t)}$ is as in the iteration procedure of Algorithm~\ref{alg:main}. With constraints $0\leq w_i\leq 1/n$ and $\sum_{i=1}^Nw_i=1$, the same result holds, except that boundary points are now those points with $w_i=0$ or $w_i=1/n$ for at least one $i$. Thus, optimal weights can be obtained by solving nonlinear equations. Normally, there is no closed-form solution, so we use a numerical approach based on Newton’s method. In general, this has the desirable quadratic convergence rate, meaning that the number of correct digits roughly doubles after each iteration \citep{Isaacson_Keller}. This alone already makes the algorithm fast, and, using the notation of the iteration step in Algorithm~\ref{alg:main}, it is applied only to the points in $\mathcal{X}_2^{(t)}$. With $n_2$ and $n_3$ as the number of points in $\mathcal{X}_2$ and $\mathcal{X}_3$, respectively, the procedure for deriving optimal weights using Newton's method is presented as Algorithm~\ref{alg:owalg}. The vector ${\boldsymbol{w}_2^{(j)}}$ in the algorithm will contain weights for the points in $\mathcal{X}_2$. During the algorithm entries can temporarily become negative (line~\ref{alg:wght} of Algorithm~\ref{alg:owalg}), but in the end all elements will be between 0 and $1/n$. The length of the vector can also change during the algorithm (line~\ref{alg:minus} of Algorithm~\ref{alg:owalg}) because the size of $\mathcal{X}_2$ changes when a weight becomes equal to 0 or $1/n$.

\begin{algorithm}
\caption{Optimal Weights Algorithm}
\label{alg:owalg} 
\begin{algorithmic}[1]
\State \textbf{Input}: The input for Algorithm~\ref{alg:main}; a bounded approximate design $\xi_\mathcal{P}$ with $0 \le w_i \le 1/n$ and the corresponding partition $\mathcal{P} = (\mathcal{X}_1, \mathcal{X}_2, \mathcal{X}_3)$ with $\mathcal{X}_2 \not = \emptyset$; a positive scalar $\epsilon_1$; the maximum number of iterations $MI_1$.
\
\State \textbf{Output}: Updated partition $\mathcal{P} = (\mathcal{X}_1, \mathcal{X}_2, \mathcal{X}_3)$ and optimal weights for the points in $\mathcal{X}_2$
\Procedure{Obtain optimal weights}{}
\State Set $n_2$ and $n_3$ as the size of $\mathcal{X}_2$ and $\mathcal{X}_3$, respectively \label{alg:start}
\State If $n_2=1$, set $\boldsymbol{w}_2 = (1-n_3/n)$ and go to line~\ref{alg:out} 
\State Else
\State \hspace{0.3cm} Set $\boldsymbol{\nu}_2^{(0)} =  \frac{1-n_3/n}{n_2} \boldsymbol{1}_{n_2-1}$, and set $\boldsymbol{w}_2^{(0)} = ({\boldsymbol{\nu}_2^{(0)}}^T, 1 - n_3/n - {\boldsymbol{\nu}_2^{(0)}}^T \boldsymbol{1}_{n_2-1})^T$ as the weight vector for points in $\mathcal{X}_2$.
\State \hspace{0.3cm} Set $\alpha = 1$ and $j=1$ \label{alg:j=1}
\State \hspace{0.3cm} Compute $\boldsymbol{\nu}_2^{(j)}=\boldsymbol{\nu}_2^{(j-1)}-\alpha\left(\frac{\partial^2 \Phi(\xi_{\mathcal{P}})}{\partial \boldsymbol{\nu}_2\boldsymbol{\nu}_2^{T}} |_{\boldsymbol{\nu}_2=\boldsymbol{\nu}_2^{(j-1)}}\right)^{-1} \left(\frac{\partial \Phi(\xi_{\mathcal{P}})}{\partial \boldsymbol{\nu}_2} |_{\boldsymbol{\nu}_2=\boldsymbol{\nu}_2^{(j-1)}} \right)$ and set $\boldsymbol{w}_2^{(j)} = ({\boldsymbol{\nu}_2^{(j)}}^T, 1 - n_3/n - {\boldsymbol{\nu}_2^{(j)}}^T \boldsymbol{1}_{n_2-1})^T$ \label{alg:wght}
\State \hspace{0.3cm} If, component-wise, $0 \le \boldsymbol{w}_2^{(j)} \le 1/n$, go to line~\ref{alg:checkeps}
\State \hspace{0.3cm} Else, set $\alpha = \alpha/2$ 
\State \hspace{0.6cm} If $\alpha > 10^{-10}$, go to line~\ref{alg:wght}
\State \hspace{0.6cm} Else, 
for $\boldsymbol{x} \in \mathcal{X}_2$ with the largest weight, if its weight exceeds $1/n$, set $\mathcal{X}_2 = \mathcal{X}_2 \backslash \{\boldsymbol{x}\}$, $\mathcal{X}_3 = \mathcal{X}_3 \cup \{\boldsymbol{x}\}$, and go to line~\ref{alg:start}; otherwise, for $\boldsymbol{x} \in \mathcal{X}_2$ with the smallest weight, set $\mathcal{X}_2 = \mathcal{X}_2 \backslash \{\boldsymbol{x}\}$, $\mathcal{X}_1 = \mathcal{X}_1 \cup \{\boldsymbol{x}\}$, and go to line~\ref{alg:start} \label{alg:minus}
\State \hspace{0.3cm} If $\lVert \frac{\partial \Phi(\xi_{\mathcal{P}})}{\partial \boldsymbol{\nu}_2} |_{\boldsymbol{\nu}_2=\boldsymbol{\nu}_2^{(j-1)}} \rVert < \epsilon_1$, set $\boldsymbol{w}_2 = \boldsymbol{w}_2^{(j)}$ and go to line~\ref{alg:out} \label{alg:checkeps}
\State \hspace{0.3cm} Else, if $j\leq MI_1$, set $j = j+1$ and go to line~\ref{alg:wght}; else set $\boldsymbol{w}_2 = \boldsymbol{w}_2^{(j)}$ and  go to line~\ref{alg:out}
\State Output $\boldsymbol{w}_2$ \label{alg:out}
\EndProcedure
\end{algorithmic}
\end{algorithm}

For the examples in Section 4, we set $MI_1$=40 and $\epsilon_1=10^{-6}$.

\subsection{An information matrix that depends on $\bt$} \label{subsec:theta} As noted previously, the information matrix $I_{\bt}$ depends on $\bt$ for GLMs. This issue has been extensively addressed in the optimal design literature. In the simplest case, one assumes that a guess for $\bt$ is available from prior experience. This value for $\bt$ can then be plugged into $I_{\bt}$, and the resulting matrix can be used in our approach to find the subdata. Alternatively, one could try to find a bounded approximate design that is not optimal for any $\bt$, but that has high efficiency for multiple plausible values of $\bt$. An approach along these lines can be found in \cite{Rios_Stufken_2025}, although it would have to be extended to bounded designs. Another approach would be to assume a prior distribution for $\bt$ that expresses the uncertainty of our knowledge of $\bt$, and use a semi-Bayesian approach to find a bounded approximate design. 

Unlike a typical design of experiments problem, for the subdata selection problem, we have a large amount of data. We can use a randomly selected subset of these data to get an idea of the value of $\bt$. Labels would have to be collected for the randomly selected subset if not yet available. We will not pursue this further here because our objective is to demonstrate the efficiency of our subdata selection method. If $\bt$ is to be estimated, that comparison will be partially obscured by the loss of efficiency resulting from estimating $\bt$.

 \subsection{Some practical considerations}
In lines~\ref{alg:newS} and \ref{alg:newP} of Algorithm~\ref{alg:main} and line~\ref{alg:minus} of Algorithm~\ref{alg:owalg}, each time only one point is exchanged. It is possible to increase the number of points that are exchanged, and this could be more efficient for some problems. However, for the problems we considered and because we have good initial subdata from IBOSS+, exchanging a single point each time is sufficient. {While applying the exchanges to the IBOSS+ subdata results in small efficiency gains, it also results in a significant increase in computation time (cf. Section~\ref{sec:empirical}).}

As discussed prior to Theorem~\ref{speed}, an exact verification of the optimality of a partition $\mathcal{P}$ with optimal weights is not possible by checking conditions (II)(b) and (II)(c) of Theorem~\ref{eq_thm}. Therefore, we will actually use the result of Theorem~\ref{speed} to find an $\epsilon$-approximation of $\xi^*$ for a small positive tolerance value $\epsilon$. For the empirical studies in Section~\ref{sec:empirical},  $\epsilon$ is set to be $10^{-6}$. As explained in Remark~\ref{rem:2}, a scale-invariant approach with a varying value for $\epsilon$ could be considered as an alternative.

Line~\ref{alg:conv} of Algorithm~\ref{alg:main} needs in our experience relatively few iterations to reach convergence. However, to safeguard against the rare situation that more iterations are needed, the maximum number of iterations can be capped by a predetermined number $MI_2$. As a practical recommendation, we suggest that $MI_2$ be set to $n$. Algorithm~\ref{alg:main} typically converges before it reaches $MI_2$, but even when it doesn't, little change occurs after increasing the number of iterations. 

\section{More simulation studies}
In this section, we present additional simulation results and comparisons, including a discussion of tuning parameter selection for the FO method, justification for excluding the PGD method, and further empirical evaluations.

\subsection{About FO method}\label{about_FO}
For the FO algorithm, the description in \cite{Ahipasaoglu_StatComp} provides no details to select an initial design or choose the maximum number of iterations, $MI_3$. 
For a fair comparison to IBOSS OBD, we start the algorithm with the design based on the IBOSS subdata and use the same convergence criterion. For $MI_3$, if it is large, there is a risk that convergence takes a long time. On the other hand, if $MI_3$ is small, the resulting design could be far from optimal. In contrast, IBOSS OBD tends to converge relatively fast. For example, for $N=1000$, $n=100$, and 100 repetitions, Table~\ref{max_iter} compares the ratio of the $A$-optimality criterion values for $\xi_{IBOSS~BOD}$ and $\xi_{FO}$ for $p=2$ and the full second-order model. Multiple values of $MI_3$ are used, keeping $MI_1 = 40$ (recall that $MI_1$ denotes the number of iterations in Algorithm~\ref{alg:owalg}). 

\begin{table}
  \centering
  \caption{Mean and standard deviation (in parentheses) of $\Phi(\xi_{\text{IBOSS OBD}})/\Phi(\xi_{\text{FO}})$ under the $A$-optimality criterion}
  \resizebox{\linewidth}{!}{%
  \begin{tabular}{lcccc}
    \hline
    $MI_3$ & $n$ & $10n$ & $20n$ & $50n$ \\ \hline
    $\Phi(\xi_{\text{IBOSS OBD}})/\Phi(\xi_{\text{FO}})$ 
      & 94.72\% (1.65\%) 
      & 99.95\% (0.03\%) 
      & 99.99\% (0.01\%) 
      & 100.00\% (0.00\%) \\ \hline
  \end{tabular}%
  }
  \label{max_iter}
\end{table}

Table~\ref{max_iter} shows that, in order to obtain results with the FO method that are comparable to those with the IBOSS OBD method, $MI_3 = n$ is insufficient, while the improvements for values of $20n$ and $50n$ are minimal compared to the results for $MI_3=10n$. We observed similar behavior for other models and criteria, so that we will $MI_3=10n$. This makes FO more competitive because a larger value of $MI_3$ comes with a significant increase in computing time, especially for $A$-optimality. The reason for this is that, even when we set $MI_3=50n$, the algorithm improves only slightly and does not reach convergence for $A$-optimality in the scenarios that we considered. This is in contrast to the performance of IBOSS OBD, which converges in all scenarios.

\subsection{About PGD method}\label{about_PGD}
The PGD method is a subset selection approach for measurement-constrained regression that targets statistical efficiency via an A-optimality criterion. It is primarily developed for linear models but can be extended to generalized linear models. Table \ref{PGD} compares the PGD method with IBOSS OBD in terms of the statistical efficiency ($A$-optimality) of the selected subsets and the corresponding computation time. The set up is as the same as that of  Scenario (i), i.e.,  A linear first-order model for $p=10$.
As shown in Table \ref{PGD}, the PGD method achieves very high statistical efficiency. However, its computational cost is relatively high, particularly for large values of $N$. Due to these scalability limitations, we excluded PGD from further consideration in subsequent analyses.
\begin{table}[ht]
\centering
\caption{$A$-optimality criterion for a first-order linear model: mean (std)}
\begin{tabular}{lcccc}
\hline
$(N, n)$ & \multicolumn{2}{c}{(1000, 100)} & \multicolumn{2}{c}{(10000, 100)} \\
\cline{2-5}
 & PGD & IBOSS OBD & PGD & IBOSS OBD \\
\hline
Eff (\%) & 98.90 (0.23) & 99.74 (0.07)& 97.61 (0.50) & 99.66 (0.08) \\
Time (s) & 2.07 (0.06)  & 0.18 (0.11) & 69.96 (1.96) & 0.71 (0.32) \\
\hline
\end{tabular}
\label{PGD}
\end{table}

\subsection{Scenario (ii): A full second-order logistic regression model with $\boldsymbol{p=3}$.}\label{add_simu_2}
A complication is that the information matrix for a logistic regression model depends on the unknown parameters. A common approach to address this is to use a locally optimal design based on a guess for the parameters, either from past experience, a pilot study, or a random sample from the full data. This challenge is the same for all methods, and we will therefore only focus on selecting efficient subdata when the true values of the parameters can be plugged in the information matrix, as is the case for simulations. We set all parameters equal to 1. Since IBOSS and OSS are designed for first-order linear models, they are not considered for this scenario.

We consider  $D$-optimality with only the parameters of the main effects are of interest. The results in  Table~\ref{table:logitD} show that SRS and LEV, while being the fastest methods, achieve efficiencies of only around 10\%. SEQ and FO perform substantially better than SRS and LEV, but remain less competitive than the proposed methods. This is expected, as they are primarily designed for settings in which all model parameters are of interest. IBOSS+ continues to achieve high efficiency while remaining more computationally efficient than SEQ and FO. Both IBOSS++ and IBOSS OBD attain near-perfect efficiencies, with IBOSS OBD converging in 99 out of 100 repetitions. 

\begin{table}
\centering
\caption{Scenario (ii): $D$-optimality criterion for the parameters of the main effects only ($N=100000$ and $n=1000$)}
\begin{tabular}{lccccccc}
\hline
Algorithm & SRS & LEV & SEQ & FO & IBOSS+ & IBOSS++ & IBOSS OBD \\
\hline
Mean Eff (\%) & 9.16 & 10.04 & 71.33 & 73.18 & 93.05 & 99.82 & 100.00 \\
Std Eff (\%)  & 0.67  & 0.60  & 0.76  & 0.56  & 1.63   & 0.26   & 0.00   \\
\hline
Mean Time (s) & 0.0002 & 0.0582 & 0.2009 & 2.3319 & 0.1341 & 0.2837 & 1.0313 \\
Std Time (s)  & 0.0000 & 0.0198 & 0.0348 & 0.2831 & 0.0176 & 0.0385 & 0.2577 \\
\hline
\multicolumn{8}{l}{\footnotesize Mean(size) = 938.54, Std(size) = 41.76; convergence rates of IBOSS OBD is 0.99.} \\
\hline
\end{tabular}%
\label{table:logitD}
\end{table}

We consider  $A$-optimality when all parameters are of interest. The SEQ method yields an average subdata size of less than 250, which is well below the target of 1000. Its efficiency is also much lower compared to that for the $D$-optimality criterion. To obtain a better comparison for the other methods, we excluded SEQ from further consideration so that we could set the subdata size at $n = 1000$ for all other methods.

SRS and LEV remain the fastest methods but have low efficiencies (less than 20\%). All remaining methods are highly efficient. IBOSS+ has the lowest efficiency among them (approximately 96\%), but is also the fastest with an average computing time of about 0.16 seconds. Both IBOSS++ and FO exhibit efficiencies close to 1, while IBOSS OBD achieves nearly perfect efficiency. IBOSS OBD converges in all repetitions, whereas FO never converges. Despite using $MI_3=10n$, this lack of convergence results in significantly longer computing times for FO than for IBOSS OBD. 

\begin{table}
\centering
\caption{Scenario (ii): $A$-optimality criterion for all parameters ($N=100000$ and $n=1000$)}

\begin{tabular}{lcccccc}
\hline
Algorithm & SRS & LEV & FO & IBOSS+ & IBOSS++ & IBOSS OBD \\
\hline
Mean Eff (\%) & 16.60 & 17.11  & 99.99 & 96.08 & 99.30 & 100.00 \\
Std Eff (\%)  & 0.88  & 0.51  & 0.00 & 1.19  & 0.02  & 0.00    \\
\hline
Mean Time (s) & 0.0002 & 0.0606 & 43.9571 & 0.1563 & 0.3221 & 1.5092  \\
Std Time (s)  & 0.0002 & 0.0100 & 1.6908 & 0.0135 & 0.0264 & 0.3608  \\
\hline
\multicolumn{7}{l}{\footnotesize Convergence rates: FO = 0.00 and IBOSS OBD = 1.00.} \\
\hline
\end{tabular}%

\label{table:logitA}
\end{table}

\subsection{Scenario (iii): Clusterwise linear model with $\boldsymbol{p=10}$ and 10 clusters.}\label{add_simu_3}
Clusterwise linear regression (CLR) models aim to uncover latent clusters within the data, such that within each cluster, the relationship between input and output variables can be adequately modeled by a linear first-order regression model. The model involves $G$ gating functions and $G$ corresponding linear regression models (also known as experts). While each observation is generated from one of these models, it is unknown from which one. CLR models have been widely applied in fields such as social sciences, environmental sciences, and engineering \citep{Brusco2003MulticriterionValue, Bagirov2017PredictionRegressionapproach, Khadka2017ComprehensiveSystems}. Research on CLR models remains active, particularly in the development of efficient algorithms to alleviate the high computational demands \citep{Park2017AlgorithmsRegression}.

Unlike standard linear regression, CLR models do not admit a closed-form expression of the maximum likelihood estimator (MLE). Due to the presence of latent cluster memberships, the Expectation-Maximization (EM) algorithm is the primary tool for estimating model parameters. 
As a result, the computational cost of fitting a CLR model can be substantial, especially for large datasets. A subdata approach can help reduce this burden. However, the absence of a closed-form expression for the information matrix for a CLR model complicates the application of optimal design techniques for subdata selection.

\cite{Liu_Stufken_Yang_2024} showed that, under certain regularity conditions, the following matrix for the $i$th data point can be used for subdata selection:
\begin{equation}
        \begin{pmatrix} \pi_1\frac{\bx_i\bx_i^T}{\sigma_1^2} & &  &\mathbf{0}\\
             &\pi_2\frac{\bx_i\bx_i^T}{\sigma_2^2} & &\\
             &&\ddots& \\
            \mathbf{0}& & & & \pi_G\frac{\bx_i\bx_i^T}{\sigma_G^2}
        \end{pmatrix},\label{infor_beta_c}
\end{equation} where $\bx_i$ is the feature vector for the $i$th data point, $\pi_g$ is the probability that the $i$th data point belongs to cluster $g$, and $\sigma_g^2$ is the variance of the error term for the $g$th linear model. The target matrix, which replaces the information matrix, is the matrix in Equation~\eqref{infor_beta_c} summed over the indices $i$ that belong to the subdata.

In our setup, we assume $G = 10$ clusters and $p=10$ features, making the dimension of the matrix in Equation~\eqref{infor_beta_c} equal to 110. Here, $\pi_g$ and $\sigma_g^2$, $g=1,\ldots,G$, are unknown parameters, just as the intercepts and slopes of the 10 linear models. As in scenario (ii), this approach requires some information about the parameters that appear in this matrix, that is, $\pi_g$s and $\sigma_g^2$s. Since this is the same for all methods, we will again use known values in our simulations, and focus on selecting efficient subdata. Specifically, we set $\sigma_g^2 = 1$ and generate $\pi_g$ values independently from a uniform distribution for $g = 1, \ldots, G$. The generated $\pi_g$ values are then normalized to ensure that $\sum_{g=1}^G \pi_g = 1$.

For this relatively more complex setting, we exclude the LEV and FO methods because they depend on a specific matrix structure that is not directly applicable here. The SEQ approach starts by selecting the first $5\times 110=550$ data points. With a target subdata size of 1000, it would only allow judicious selection of 450 data points. Therefore, to give SEQ a better chance to shine, we consider a subdata size of 10000 in this scenario. 

Under the $D$-optimality criterion in Table~\ref{CLR_D1}, SRS achieves an efficiency of 53\% and is the fastest method. SEQ performs well, with an efficiency of 97\%. IBOSS+ outperforms SEQ, achieving near-perfect efficiency while requiring less computing time. Both IBOSS++ and IBOSS OBD exhibit nearly perfect efficiency, but at the cost of significantly longer computing times. In particular, IBOSS OBD converges in all repetitions, which ensures high average efficiency. Additionally, we observe that the average size of the selected subdata is less than 8000, which is considerably below the target size of 10000.

\begin{table}
\centering
\caption{Scenario (iii): $D$-optimality criterion for all parameters ($N=100000$ and $n=10000$)}
\begin{tabular}{lccccc}
\hline
Algorithm & SRS & SEQ & IBOSS+ & IBOSS++ & IBOSS OBD \\
\hline
Mean Eff (\%) & 53.39 & 97.03 & 99.87 & 100.00 & 100.00 \\
Std Eff (\%)  & 0.26  & 0.13  & 0.02  & 0.00   & 0.00   \\
\hline
Mean Time (s) & 0.0002 & 2.4255 & 1.7737 & 89.8514 & 175.9158 \\
Std Time (s)  & 0.0000 & 0.1496 & 0.2002 & 5.3997 & 9.7151 \\
\hline
\multicolumn{6}{l}{\footnotesize Mean(size) = 7881.3, Std(size) = 104.39; convergence rate of IBOSS OBD = 1.00.} \\
\hline
\end{tabular}%
\label{CLR_D1}
\end{table}

For the $A$-optimality criterion, just as in Scenario (ii), the average subdata size selected by the SEQ approach falls at less than 1000 far below the target of 10000. As a result, SEQ is also very inefficient for the same reason as explained in the previous paragraph. Consequently, SEQ is excluded from further consideration, and the subdata size is fixed at 10000 for the other methods.

Table~\ref{CLR_A1} shows that the performance of all methods is similar to that for the $D$-optimality criterion. In particular, IBOSS+ maintains high efficiency while requiring significantly less computing time than IBOSS++ and IBOSS OBD, which achieve nearly perfect efficiencies. IBOSS OBD again converges in all repetitions, ensuring the reliability of its results.

\begin{table}
\centering
\caption{Scenario (iii): $A$-optimality criterion for all parameters ($N=100000$ and $n=10000$))}
\begin{tabular}{lcccc}
\hline
Algorithm & SRS & IBOSS+ & IBOSS++ & IBOSS OBD \\
\hline
Mean Eff (\%) & 53.58 & 99.41 & 100.00 & 100.00 \\
Std Eff (\%)  & 0.24  & 0.03  & 0.00  & 0.00  \\
\hline
Mean Time (s) & 0.0003 & 1.9989 & 154.8430 & 305.8149 \\
Std Time (s)  & 0.0001 & 0.1740 & 9.7072 & 13.3035  \\
\hline
\multicolumn{5}{l}{\footnotesize Convergence rate of IBOSS OBD = 1.00} \\
\hline
\end{tabular}%
\label{CLR_A1}
\end{table}

\end{document}